\documentclass{siamltex}
\usepackage{bm}
\usepackage{bbm}
\usepackage{graphicx}
\usepackage{url}
\usepackage{color}
\usepackage{mathtools}
\usepackage{amsmath}	
\usepackage{amssymb}	
\usepackage{graphicx,graphics,color}	
\usepackage{enumerate}
\usepackage{cancel}
\usepackage{color}

\newtheorem{thm}{Theorem}[section] 
\newtheorem{remark}[thm]{Remark}

\usepackage{algorithmic}	
\usepackage[section]{algorithm}

\title{Synchronization of heterogeneous oscillators under network modifications: Perturbation and optimization of the synchrony alignment function
\thanks{
This work was funded in part by NSF Grant No.~DMS-1127914 (DT), NIH Grant No.~R01HD075712 (DT), the James S. McDonnell Foundation No.~220020325 (PSS),  {ARO Grants No.~W911NF-12-1-0276 (JS) and No.~W911NF-16-1-0081 (JS),} and Simons Foundation Grant No.~318812 (JS). Any opinions, findings, and conclusions or recommendations expressed in this material are those of the author(s) and do not necessarily reflect the views of the funding agencies. 
}}
\author{
Dane Taylor\thanks{Carolina Center for Interdisciplinary Applied Mathematics, Department of Mathematics, University of North Carolina, Chapel Hill, NC 27599, USA;
and Statistical and Applied Mathematical Sciences Institute (SAMSI), Research Triangle Park, NC, 27709, USA}
\and Per Sebastian Skardal\thanks{Department of Mathematics, Trinity College, Hartford, CT, 06106, USA}
\and Jie Sun\thanks{Department of Mathematics, Clarkson University, Potsdam, NY, 13699, USA; Department of Physics, Potsdam, NY, 13699, USA;  {Department of Computer Science, Potsdam, NY, 13699, USA.}}
}

\begin{document}

\maketitle

\begin{abstract}
Synchronization is central to many complex systems in engineering physics (e.g., the power-grid, Josephson junction circuits, and electro-chemical oscillators) and biology (e.g., neuronal, circadian, and cardiac rhythms). Despite these widespread applications---for which proper functionality depends sensitively on the extent of synchronization---there remains a lack of understanding for how systems  {can best} evolve and adapt to enhance or inhibit synchronization. We study how network modifications affect the synchronization properties of network-coupled dynamical systems that have heterogeneous node dynamics (e.g., phase oscillators with non-identical frequencies), which is often the case for real-world systems. Our approach relies on a \emph{synchrony alignment function} (SAF) that quantifies the interplay between heterogeneity of the network and of the oscillators and provides an objective  measure for a system's ability to synchronize. We conduct a spectral perturbation analysis of the SAF for structural network modifications including the addition and removal of edges, which subsequently ranks the edges according to their importance to synchronization. Based on this analysis, we develop gradient-descent algorithms to efficiently solve optimization problems that aim to maximize phase synchronization via network modifications. We support these and other results with numerical experiments.
\end{abstract}



\pagestyle{myheadings}
\thispagestyle{plain}
\markboth{D. TAYLOR {\it et al.}}{Phase synchronization under network modifications}

\section{Introduction}
The study of synchronization is a multidisciplinary pursuit \cite{glass1988clocks,pikovsky2003synchronization,arenas2008synchronization} aimed to understand how dynamics occurring for individual oscillators (which can represent a wide array of phenomena ranging from populations of firing neurons to generators in a power grid \cite{dorfler2012synchronization,nishikawa2015comparative,rohden2012self,skardal2015control}) can combine so that the system exhibits self-organized, collective behavior. For numerous systems, proper functionality requires an appropriate amount of synchronization. The power grid, for example, must provide electricity following regional specifications (e.g., alternating current at 120 volts and 60 hertz in the United States) and a breakdown of synchronization can lead to costly blackouts \cite{parrilo1999model,simonsen2008transient,susuki2009global,motter2013spontaneous}. Other technologies in which synchronization plays a crucial role include Josephson junctions circuits \cite{wiesenfeld1996synchronization,rosenblum2007self}, physical infrastructure \cite{strogatz2005theoretical}, electro-chemical oscillators \cite{kazanci2007pattern}, synthetic biological oscillators \cite{prindle2012sensing}, and distributed sensor networks \cite{olfati2007consensus,sarlette2009consensus,olshevsky2011convergence,olshevsky2009convergence}. Synchronization is also ubiquitous in biological systems \cite{winfree1967biological}, where applications include coordinated neuronal activity in the brain \cite{medvedev2001synchronization,schnitzler2005normal}, cardiac rhythms of the heart \cite{mirollo1990synchronization,karma2013physics}, circadian rhythms governing sleep cycles \cite{saper2005hypothalamic}, gene regulation \cite{kuznetsov2004synchrony}, and intestinal activity \cite{ermentrout1984frequency,aliev2000simple}. Excess synchronization in the brain, for example, has been linked to tremors and seizures \cite{schnitzler2005normal,wilson2015clustered}.

\begin{figure}[t]
\centering
\includegraphics[width=\linewidth]{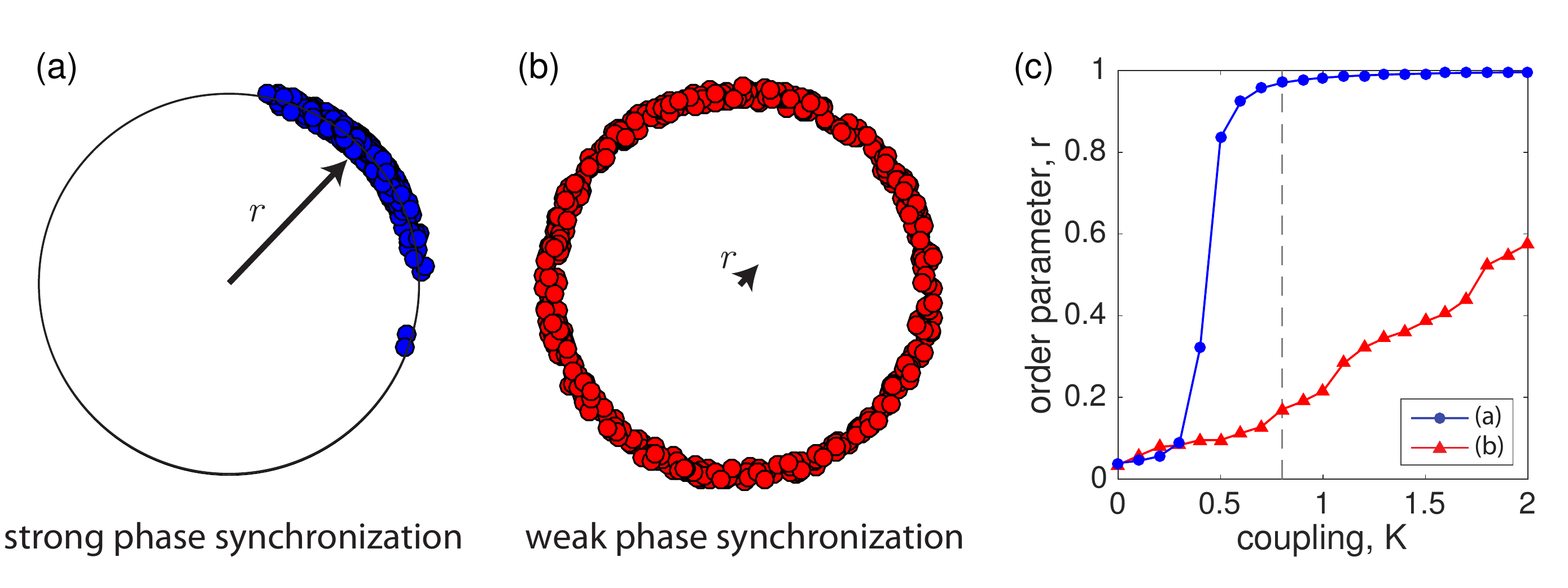}
\caption{ 
Phase synchronization depends  {crucially} on the alignment of heterogeneous oscillator dynamics (i.e., as indicated by their natural frequencies $\{\omega_n\}$) with the heterogeneity of the network structure  {(which is manifest in the eigenvalues and eigenvectors of the network Laplacian matrix $L$).}
(a),(b)~Phase-locked oscillators $\{\theta_n\}$ (shown here embedded on the unit circle) for states of strong ($r\approx1$) and weak ($r\approx0$) phase synchronization, respectively.  {Here, $r$ is the Kuramoto order parameter given by Eq.~\eqref{eq:OrderParameter}.} These simulations reflect phase synchronization of the Kuramoto model [Eq.~\eqref{eq:Kuramoto1} and $H(\theta)=\sin(\theta)$] with coupling strength $K=0.8$ and network coupling given by the Erd\H{o}s-R\'enyi (ER) model \cite{erdos1960evolution} with $N=500$ nodes, mean degree 4, and minimum degree of $d_{min}=2$. The only difference between the systems studied in panels (a) and (b) is how the natural frequencies align with the network structure;  {panels (a) and (b) correspond to maximizing and minimizing phase synchronization, respectively (in the notation introduced in Sec.~\ref{sec:Derive}, these correspond to $\omega_n = v_n^{(N)}$ and $\omega_n = v_n^{(2)}$, respectively, where  $\bm{v}^{(m)}$ is the eigenvector corresponding to the $m$-th smallest eigenvalue of the network Laplacian).}
(c)~Dependence of $r$ on $K$ for these two systems. The vertical dashed line indicates the value of  $K$ used to produce panels (a) and (b). See Sec.~\ref{sec:saf_example} for further discussion  {of the simulation.}
}
\label{fig:Kuramoto_toy}
\end{figure}

Given these widespread applications, it is important to develop theory to control, engineer and optimize the synchronization properties of complex systems---particularly, heterogeneous systems. In this research, we explore what we believe to be one of the most fundamental pursuits in this direction, understanding the effect of a network modification such as the addition or removal of an edge or set of edges  {on phase synchronization}. This fundamental topic has been previously studied for complete (perfect) synchronization of identical oscillators \cite{barahona2002synchronization,milanese2010approximating,dadashi2010rewiring,jalili2013enhancing} (i.e., based on the Master Stability  {Function} \cite{pecora1998master}) and nonidentical oscillators in the weak synchronization regime \cite{restrepo2006characterizing,milanese2010approximating,taylor2011network,taylor2012social} (i.e., the onset of synchronization \cite{restrepo2005onset,restrepo2006emergence}). We develop theory for phase synchronization of nonidentical oscillators in the strong synchronization regime, thereby filling an important gap in the established literature.

Our approach relies on a \emph{synchrony alignment function} (SAF) \cite{skardal2014optimal} that  {quantifies} the interplay between heterogeneity in the network and heterogeneity of the oscillators and provides insight into a network's ability to synchronization. We showed in \cite{skardal2014optimal} that minimization of the SAF gives a maximization of phase synchronization, and we developed greedy, Monte-Carlo algorithms to optimize the phase synchronization of networks under various constraints. See Fig.~\ref{fig:Kuramoto_toy} for a numerical experiment highlighting the effectiveness of this approach. Because this approach is based on a mathematical analysis, it is much more reliable than---yet in agreement with---known heuristics for enhancing synchronization such as implementing  {negative} correlations between the frequencies of neighboring oscillators \cite{brede2008synchrony,buzna2009synchronization,skardal2014optimal} or incorporating  {positive} correlations between the oscillators' degrees and  {natural frequency magnitudes} \cite{brede2008synchrony,skardal2014optimal}.  In addition to optimization, the SAF can be used to explore fundamental limitations on phase synchronization for systems with frustrated coupling---a phenomenon referred to as the \emph{erosion of synchronization} \cite{skardal2015erosion_B,skardal2015erosion}.  {In continuing to develop this theoretical framework, we} recently generalized the SAF to directed networks \cite{skardal2016optimal}.

Here, we conduct a spectral perturbation analysis of the SAF to analyze the effect on phase synchronization due to structural network modifications. 
This analysis ranks the edges (and potential edges) according to their importance to synchronization. Importantly, this ranking (i.e., centrality measure \cite{taylor2015eigenvector}) takes into account the full system---that is, both the particular network structure and the oscillators' (potentially) heterogeneous natural frequencies and is akin to other rankings that are specific to a particular class of dynamics \cite{restrepo2006characterizing,gleich2015pagerank,skardal2015collective}. Moreover, we study a class of optimization problem in which  the goal is to maximally enhance phase synchronization through the addition and removal of a fixed numbers of edges. Using these rankings, we develop efficient gradient-descent algorithms to yield approximate solutions. We support these and other findings with numerical experiments. 

The remainder of this paper is organized as follows. In Sec.~\ref{sec:model}, we introduce the oscillator models that we study and order parameters to quantify phase synchronization. In Sec.~\ref{sec:SAF}, we present the SAF, derive its upper and lower bounds,  {and describe two pedagogical network examples.} In Sec.~\ref{sec:Perturb}, we present a spectral perturbation analysis of the SAF for a system undergoing a network modification. In Sec.~\ref{sec:OptimalSection}, we  {present the ranking of} edges according to their importance to phase synchronization.  {In Sec.~\ref{sec:algorithm_sec}, we} develop gradient-descent algorithms to efficiently enhance synchronization. We provide a discussion in Sec.~\ref{sec:Discussion}.

\section{Oscillator Models for Phase Synchronization}\label{sec:model}
We define in Sec.~\ref{sec:models} two related models that exhibit phase synchronization, the nonlinear \emph{Kuramoto phase-reduction model} \cite{kuramoto2012chemical} and the linear \emph{heterogeneous Laplacian dynamics (HLD)}. As we showed in \cite{skardal2014optimal}, the linear HLD approximates the synchronization of nonlinear systems in the regime of strong phase synchronization. To quantify the extent of phase synchronization of both systems, in Sec.~\ref{sec:phase_sync} we define two order parameters, the \emph{Kuramoto order parameter} $r$ and \emph{variance order parameter} $R$, and show that they are approximately equal in the strong synchronization regime.

\subsection{Oscillator Models}\label{sec:models}
We first define Kuramoto's model for weakly coupled limit-cycle oscillators.
\begin{definition}[Kuramoto Phase-Reduction Model \cite{kuramoto2012chemical}]
Consider $N$ phase oscillators in which $\theta_n\in[0,2\pi)$ is the phase of oscillator $n$, $\hat{\omega}_n\in\mathbb{R}$ is the natural frequency of oscillator $n$, matrix $\hat{A}_{nm}$ encodes the network-coupling of oscillators, and $H_{nm} :(-\pi,\pi)\to\mathbb{R}$ is an interaction-specific,  {$2\pi$-periodic coupling function that is differentiable at $0$.} The Kuramoto phase-reduction model \cite{izhikevich2008phase} is given by the system of first-order nonlinear differential equations 
\begin{align}
\frac{d{\theta}_n}{dt}= \hat{\omega}_n + {K}\sum_{m=1}^N \hat{A}_{nm}H_{nm}(\theta_m-\theta_n),~~~n\in\{1,\dots,N\}.
\label{eq:Kuramoto1}
\end{align}
\end{definition}

Kuramoto derived Eq.~\eqref{eq:Kuramoto1} as a phase-reduction model \cite{izhikevich2008phase} to describe the synchronization of weakly interacting limit cycle oscillators (i.e., the coupling is sufficiently weak so that the limit cycles are not destroyed). Often, it is assumed that the oscillator interactions follow an identical functional form, $H_{nm}(\theta)=H(\theta)$. Under the choice $H(\theta)=\sin(\theta)$, which represents the first-order term of a Fourier expansion for an odd function $H(\theta)$, Eq.~\eqref{eq:Kuramoto1} is widely referred to simply as the ``Kuramoto model,'' and it is one of the most paradigmatic nonlinear systems for the study of synchronization. 
It has been used to study, for example, the power grid \cite{dorfler2012synchronization,nishikawa2015comparative,skardal2015control}, animal movements \cite{nabet2009dynamics}, clapping audiences \cite{taylor2010spontaneous} and many more applications \cite{acebron2005kuramoto,arenas2008synchronization,pikovsky2003synchronization}.

We also study synchronization according to the  following linear system.

\begin{definition}[Heterogeneous Laplacian Dynamics]
Consider $N$ oscillators with phases $\{\theta_n\}$ and natural frequencies $\{\omega_n\}$ that are coupled by a network given with adjacency matrix $A$, where $A_{nm}$ encodes the impact of oscillator $m$ on oscillator $n$. Letting $L_{nm} = -A_{nm}+ \delta_{nm}\sum_{m}A_{nm}$ define the combinatorial Laplacian matrix corresponding to $A$, the system is given for $n\in\{1,\dots,N\}$ by  
\begin{align}
\frac{d{\theta}_n}{dt} = {\omega}_n-K\sum_{m=1}^N {L}_{nm}\theta_m , \label{eq:HLD}
\end{align}
which can be written in matrix form by $d{\bm{\theta}}/dt = \bm{\omega}-K{L}\bm{\theta}$.
\end{definition}

In previous research \cite{skardal2014optimal,skardal2016optimal}, we showed in the regime of strong phase synchronization that the dynamics of Eq.~\eqref{eq:Kuramoto1} can be approximated by Eq.~\eqref{eq:HLD}. In particular, if one defines {$\omega_n=\hat{\omega}_n+K\sum_{m}\hat{A}_{nm}H_{nm}(0)$ and $A_{nm}=\hat{A}_{nm}H'_{nm}(0)$,}
then Eq.~\eqref{eq:HLD} gives the linearization of Eq.~\eqref{eq:Kuramoto1} around the synchronization manifold \cite{skardal2014optimal,skardal2016optimal}. For example, phase-locked solutions of Eq.~\eqref{eq:HLD} approximate phase-locked solutions of Eq.~\eqref{eq:Kuramoto1}. In addition to providing insight into the synchronization of nonlinear systems, we note that Eq.~\eqref{eq:HLD} has many applications itself including consensus algorithms for sensor networks \cite{olfati2007consensus,sarlette2009consensus,olshevsky2011convergence}, where it is often assumed that $\omega_n=\omega$ for each $n$.

\subsection{Quantifying Phase Synchronization}\label{sec:phase_sync}
Many notions of synchronization have been studied, {each capturing different physical characteristics of real-world systems}. For identical oscillators (i.e., those in which $\hat{\omega}_n=\hat{\omega}$ or $ {\omega}_n= {\omega}$ for every $n$), one often studies whether the oscillators obtain \emph{perfect} phase synchronization, whereby all phases converge so that $\lim_{t\to\infty}|\theta_n(t)-\theta_m(t)|=0$. For systems with heterogeneous dynamics, such as when $\{\omega_n\}$ or $\{\hat{\omega}_n\}$ are non-identical (which is typical in real-world scenarios), this notion of synchronization is too restrictive  \cite{sun2009master}. Here, we study states in which the phase oscillators are \emph{phase-locked} and the oscillators achieve strong phase synchronization. That is, for any oscillators $n$ and $m$ the phase difference $\theta_n(t)-\theta_m(t)$ is assumed  {to relax to a small, constant value} $|\theta_n(t)-\theta_m(t)|\ll1$. We note that phase locking implies perfect frequency synchronization so that $ {d\theta_n}/{dt}= {d\theta_m}/{dt}=\Omega$ for any pair of nodes $n$ and $m$, where $\Omega=N^{-1}\sum_n\omega_n$ \cite{skardal2015collective} is the collective frequency for undirected networks. 

Because phase-locked oscillators need not converge---instead, they cluster around some central phase, or a \emph{mean field}---it is important to measure (quantify) the extent of phase synchronization. To this end, we study two measures of phase synchronization, the Kuramoto order parameter, $r$, and the variance order parameter, $R$,  {to be defined below.}  We note that $r$ is the most common for Eq.~\eqref{eq:Kuramoto1}; however, for analytical purposes, it is advantageous to measure phase synchronization based on $R$. In principle, either order parameter ($r$ or $R$) can be applied to either system [Eq.~\eqref{eq:Kuramoto1} or Eq.~\eqref{eq:HLD}], and as we shall show, the order parameters are approximately equal in the strong synchronization regime.

\begin{definition}[Kuramoto Order Parameter \cite{kuramoto2012chemical}]
Given a system of coupled oscillators with phases $\{\theta_n\}$ [e.g., Eq.~\eqref{eq:Kuramoto1} or Eq.~\eqref{eq:HLD}], the Kuramoto order parameter $r$ and mean field $\psi$ are found by mapping the phases onto the unit circle and calculating the centroid,
\begin{align}
r e^{i\psi}= \frac{1}{N}\sum_{n=1}^N e^{i\theta_n} \label{eq:OrderParameter},
\end{align}
where $r\geq0$ and $\psi\in[0,2\pi)$.
\end{definition}
\begin{remark}
By definition, the value $r\in[0,1]$.
Importantly, $r\approx1$ indicates strong phase synchronization, whereas $r\approx0$  {typically} indicates weak (or a lack of) phase synchronization. See Fig.~\ref{fig:Kuramoto_toy}(a) and (b) for illustrations of these two cases.
\end{remark}

\begin{definition}[Variance Order Parameter]
Given a system of coupled oscillators with phases $\{\theta_n\}$ [e.g., Eq.~\eqref{eq:Kuramoto1} or Eq.~\eqref{eq:HLD}], 
we define 
\begin{align}
R = 1 -  {\sigma_\theta^2}/{2} \label{eq:varianceOrderParameter}.
\end{align}
where $\sigma_\theta^2 = {N}^{-1} \sum_{n} ({\theta_n}-\overline{\theta} )^2 =N^{-1}||\bm{\theta}-\overline{\theta}\bm{1}||_2^2$ is the variance of phases and the mean phase $\overline{\theta}=N^{-1}\sum_n\theta_n$ defines a mean field.
\end{definition}

Order parameters $r$ and $R$ both limit to unity for perfect synchronization, and ``strong synchronization'' is defined as the regime in which $r\approx R\approx 1$. We now establish that these order parameters are approximately equal in this regime through the following bounds.

\begin{proposition}[Equivalence of Order Parameters]\label{lemma1}
Assume that the infinite sequence $\{\|\bm{\theta}-\psi \bm 1\|^k_k/k!\}$ for $k\in\{2,4,\dots\}$ monotonically converges to zero so that 
\begin{align}
\lim_{k\to\infty}\frac{\|\bm{\theta}-\psi \bm 1\|_k^k}{k!}\to0 ,
\end{align}
where $||\cdot ||_p$ denotes the $p$-norm, and
\begin{align}
\frac{\|\bm{\theta}-\psi\|_2^2}{2!}& {\ge\frac{\|\bm{\theta}-\psi\|_3^3}{3!}\ge\cdots\ge} \frac{\|\bm{\theta}-\psi\|_k^k}{k!}>\cdots, 
\end{align}
then Eqs.~\eqref{eq:varianceOrderParameter} and \eqref{eq:OrderParameter} satisfy the following bounds,
\begin{align}
R - \frac{|\overline{\theta}-\psi|^2}{2} \le  r \le  R + \frac{||\bm \theta - \psi \bm 1||_4^4}{24N} \label{eq:equivalence}.
\end{align}
Moreover, the difference between the two mean fields, $\psi$ and $\overline{\theta}$, is bounded by
\begin{align}
|\overline{\theta}-\psi|  &\le \frac{||\bm \theta-\psi \bm1||_3^{3} }{6N}.\label{eq:meanfields}
\end{align}
\begin{proof}
See Appendix~\ref{appendixA}.
\end{proof}
\end{proposition}

As we show in Appendix \ref{appendixA}, the variance order parameter $R$ captures the leading order term of an expansion of $r$ near $r=1$, and the upper and lower bounds in Eq.~\eqref{eq:OrderParameter} come from the next terms in the expansion.   Both $\frac{|\overline{\theta}-\psi|^2}{2}$ and $\frac{||\bm \theta - \psi \bm 1||_4^4}{24N}$ become vanishingly small in the strong synchronization regime, implying that $r\approx R$ is a valid and accurate approximation in this regime.

\section{The Synchrony Alignment Function (SAF)}\label{sec:SAF}
We now present a derivation of the SAF, which quantifies the ability for a heterogeneous system to synchronize by measuring the alignment of the heterogeneity of the nodal dynamics (e.g., oscillators' natural frequencies) with that of the network (as measured through the spectral properties of the Laplacian matrix).
In Sec.~\ref{sec:Derive}, we present the SAF and its connection to phase synchronization. 
In Sec.~\ref{sec:Bound}, we develop upper and lower bounds on the SAF.
In Sec.~\ref{sec:toy_examples}, we study these bounds for two pedagogical network examples.
In Sec.~\ref{sec:saf_example}, we describe a numerical experiment to highlight the  {applicability of using SAF for} optimizing phase synchronization.

\subsection{Phase Synchronization and the SAF}\label{sec:Derive}
 {A main advantage of order parameter $R$ versus $r$ for HLD systems is that $R$ can be solved exactly in terms of the SAF. 
Herein, we obtain a solution $\bm{\theta}^* $ for the phase-locked state of HLD systems given by  Eq.~\eqref{eq:HLD}. Using this solution, we obtain an analytical expression for $R$, which can be succinctly expressed in terms of the SAF. 
 
 We first present a solution to the phase-locked state of HLD systems.}  

\begin{theorem}[Phase-locked State of Heterogeneous Laplacian Dynamics \cite{skardal2014optimal}]
Consider the  {Heterogeneous Laplacian Dynamics given by Eq.~\eqref{eq:HLD}, for which we assume $L$ describes a connected, undirected network, and}
let ${L}^\dagger=\sum_{n=2}^N \lambda_n\bm{v}^{(n)}{\bm{v}^{(n)}}^T$ denote the Moore-Penrose pseudo-inverse \cite{ben2003generalized} of the Laplacian matrix $L$. Then the equilibrium (i.e., phase-locked) solution is given by
\begin{align}
\bm{\theta}^* =  K^{-1}{L}^\dagger \bm{ {\omega}}  + \overline{\theta}\bm{1}\label{eq:stationarySolution},
\end{align}
and the variance order parameter $R$ is given by
\begin{align}
R = 1-J(\bm{ {\omega}},  L)/2K^2 ,\label{eq:rObjective}
\end{align}
 {where $J(\bm{ {\omega}},  L)$ is the synchrony alignment function defined below.}
\begin{proof}
See Appendix~\ref{appendixB}
\end{proof}
\end{theorem}

\begin{definition}[Synchrony Alignment Function  {(SAF)} for Undirected Networks \cite{skardal2014optimal}]
Let $\bm{\omega}$ denote a vector encoding oscillators' natural frequencies and consider an undirected network with Laplacian $L$ having
eigenvalues $0=\lambda_1<\lambda_2\le \lambda_3\le \dots\le\lambda_N$ and corresponding eigenvectors $\{\bm{v}^{(n)}\}$. Let ${L}^\dagger=\sum_{n=2}^N \lambda_n\bm{v}^{(n)}{\bm{v}^{(n)}}^T$ denote the Moore-Penrose pseudo-inverse \cite{ben2003generalized} of $L$. We define the SAF by
\begin{align}
J(\bm{ {\omega}},  L)=N^{-1}|| L^\dagger \bm{\omega}||^2_2 = \frac{1}{N}\sum_{n=2}^N \frac{ ( \bm{\omega}^T\bm{v}^{(n)})^2}{\lambda_n^{2}.\label{eq:Objective}}
\end{align}
\end{definition}

\begin{remark}
Given that the eigenvectors $\{\bm{v}^{(n)}\}$ of $L$ form an orthonormal basis for $\mathbb{R}^N$ and that the terms in the summation of Eq.~\eqref{eq:Objective} are proportional to $1/\lambda_n^2$, the SAF will be smaller (larger) if the frequency vector $\bm{\omega}$ is more strongly aligned with eigenvectors corresponding to large (small) eigenvalues.
\end{remark}

\subsection{Bounding the SAF}\label{sec:Bound}
Equation~\eqref{eq:rObjective} highlights for HLD systems that $R$ can be solved in terms of the SAF, which is advantageous for the optimization of phase synchronization through tuning $R$ (which approximates $r$ in the strong synchronization regime). 
We now develop upper and lower bounds on the SAF and use them to solve the optimization problems of maximizing and minimizing $R$ for a fixed network and natural frequencies with mean $\overline{\omega}=\sum_n\omega_n$ and specified variance $\sigma^2_\omega = N^{-1}\sum_n(\omega_n-\overline{\omega})^2$.

\begin{proposition}[Bounding the SAF  \cite{skardal2014optimal}] 
Consider the SAF given by Eq.~\eqref{eq:Objective}, where the oscillators have natural frequencies with variance $\sigma^2_\omega$ and $L$ denotes the {Laplacian} of an undirected, connected network. The SAF satisfies  
\begin{align}
 \frac{\sigma_\omega^2}{N\lambda_N^2}\le J(\bm{{\omega}},L)\le \frac{\sigma_\omega^2}{N\lambda_2^2} .
\label{eq:Bound1}
\end{align}
\begin{proof}
Recall that the eigenvectors $\{\bm{v}^{(n)}\}$ form an orthonormal basis for $\mathbb{R}^N$. It follows that the frequency vector can be expressed as $\bm{ {\omega}}=\sum_n \alpha_n \bm{v}^{(n)}$ where components $\alpha_n$ are given by $\alpha_n =  \bm{\omega}^T\bm{v}^{(n)}$. After substituting this into Eq.~\eqref{eq:Objective}, we find $J(\bm{ {\omega}},  L) =  N^{-1}\sum_{n=2}^N  { \alpha_n^2}/{\lambda_n^{2}}.$  Note also that $\{\alpha_n\}$ must satisfy the constraint $\sigma_\omega^2 = \sum_{n=2}^N \alpha_n^2$.  We obtain the left-hand inequality by using $\lambda_N^{-2} \le \lambda_n^{-2}$ for any $n$.  We obtain the right-hand inequality by using $\lambda_2^{-2} \ge \lambda_n^{-2}$ for any $n$.
\end{proof}
\end{proposition}

\begin{corollary}
The maximization and minimization of Eq.~\eqref{eq:Objective} for fixed $L$ over the space of natural frequencies $\{\bm{\omega}:\overline{\omega}=\sum_n\omega_n \text{ and } N^{-1}\sum_n(\omega_n-\overline{\omega})^2=\sigma_\omega^2\}$ have the solutions $\bm{\omega}=\overline{\omega} \pm\sigma_\omega\bm{v}^{(2)}$ and $\bm{\omega}=\overline{\omega} \pm\sigma_\omega\bm{v}^{(N)}$, respectively.
\begin{proof}
Substitution of $\bm{\omega}=\overline{\omega} \pm\sigma_\omega\bm{v}^{(2)}$ into Eq.~\eqref{eq:Objective} recovers the upper bound, whereas substitution of $\bm{\omega}=\overline{\omega} \pm\sigma_\omega\bm{v}^{(N)}$ into Eq.~\eqref{eq:Objective} recovers the lower bound.
\end{proof}
\end{corollary}

\begin{corollary}\label{co:maxmin}
Considering the system in Eq.~\eqref{eq:HLD}, the maximization and minimization of $R$ given by Eq.~\eqref{eq:varianceOrderParameter} over the space of natural frequencies $\{\bm{\omega}:\overline{\omega}=\sum_n\omega_n \text{ and }N^{-1}\sum_n(\omega_n-\overline{\omega})^2=\sigma_\omega^2\}$ for fixed $L$ have the solutions $\bm{\omega}=\overline{\omega}\pm \sigma_\omega\bm{v}^{(N)}$ and $\bm{\omega}=\overline{\omega}\pm \sigma_\omega\bm{v}^{(2)}$, respectively.
\begin{proof}
From Eq.~\eqref{eq:rObjective}, we can see that $R$ is a linear function  {of $J(\bm{\omega},L)$} so that $\text{argmax}_{\bm{\omega}} R=\text{argmin}_{\bm{\omega}}J(\bm{\omega},L)$ and $\text{argmin}_{\bm{\omega}} R=\text{argmax}_{\bm{\omega}}J(\bm{\omega},L)$.
\end{proof}
\end{corollary}

\begin{remark}
Given the equivalence relation defined in Eq.~\eqref{eq:equivalence}, the maximization of $R$ approximates the maximization of $r$, which is expected to be accurate in the regime of strong synchronization.
\end{remark}

\subsection{SAF for Pedagogical Network Examples}\label{sec:toy_examples}

To provide intuition toward synchrony optimization with the SAF, in this section we study the maximization and minimization of $R$ using the SAF for two pedagogical networks---an undirected chain and a star network.

We first consider an undirected chain, which is shown in Fig.~\ref{fig:toy}(a),(b) and is a network consisting of sequentially linked nodes with end nodes indexed $n=1$ and $N$. The Laplacian matrix for a chain takes the form
\begin{align}
L^{(chain)} = \begin{bmatrix}
1 & -1 & \ldots & 0 & 0 \\
-1 & 2 & \ldots & 0 & 0 \\
\vdots & \vdots & \ddots & \vdots & \vdots \\
0 & 0 & \ldots & 2 & -1 \\
0 & 0 & \ldots & -1 & 1
\end{bmatrix},\label{eq:Chain01}
\end{align}
and has eigenvalues 
\begin{align}
\lambda_n  = 4\sin^2\left(\frac{\pi(n-1)}{2N}\right)
\label{eq:chain_lambda}
\end{align}
and corresponding eigenvectors $\{\bm{v}^{(n)}\}$ with entries 
\begin{align}
 {v}^{(n)}_m = \left\{ 
\begin{array}{rl}
\frac{1}{\sqrt{N}},&n=1\\
\sqrt{\frac{2}{N}}\cos\left(\frac{\pi(n-1)(2m-1)}{2N}\right),&n\ge2 .
\end{array}
\right.
\label{eq:chain_v}
\end{align}
We depict the eigenvectors $\{\bm{v}^{(n)}\}$ for $n\ge2$ in Fig.~\ref{fig:toy}(c). 
It follows that the SAF obtains a minimum value 
\begin{align}
{\min_{\|\bm{\omega}\|=1} J(\bm{\omega},L^{(chain)}) = \frac{1}{N\lambda_N^2}  = \frac{1}{16N\sin^{4}(\pi(N-1)/2N)}}
\end{align}
when $\bm{\omega} = \bm{v}^{(N)}$ and a maximum value 
\begin{align}
{\max_{\|\bm{\omega}\|=1} J(\bm{\omega},L^{(chain)})=  \frac{1}{N\lambda_2^2} = \frac{1}{16N\sin^{4}(\pi/2N)}}
\end{align}
when $\bm{\omega} = \bm{v}^{(2)}$. Recall that the maximization of $R$ corresponds to minimization of the SAF, and vice versa.

\begin{figure}[t]
\centering
\includegraphics[width=1\linewidth]{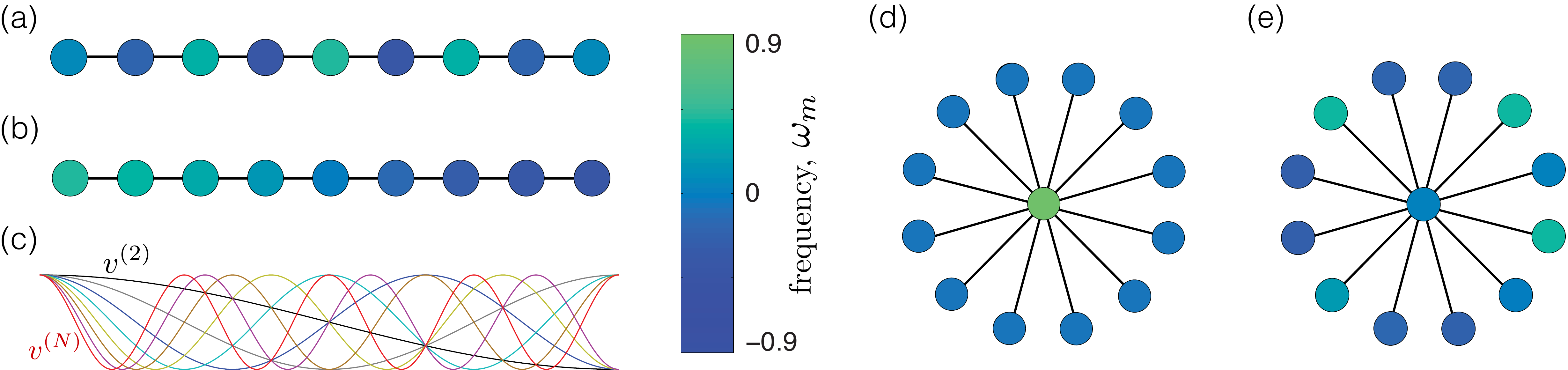}
\caption{
Pedagogical network examples including (a)--(c) a chain network with $N=9$ nodes and (d)--(e) a star network with $N=13$. The nodes' colors indicate the optimal natural frequency $\omega_m$ for each node $m$ that either maximizes $R$ (i.e., $\bm{\omega}=\bm{v}^{(N)}$), which is shown in panels (a) and (d), or minimizes $R$ (i.e., $\bm{\omega}=\bm{v}^{(2)}$), which is shown in panels (b) and (e). Panel (c) depicts the eigenvectors $\{\bm{v}^{(n)}\}$ for the chain.
}
\label{fig:toy}
\end{figure}

We next consider the star network shown in Fig.~\ref{fig:toy}(c),(d) in which there is a central \emph{hub} node with degree $d_1=N-1$ and is connected to \emph{leaf} nodes of degree $d_n=1$ for $n\ge2$. The network Laplacian matrix is given by
\begin{align}
L^{(star)} = \begin{bmatrix}
N-1 & -1 & -1 & \ldots & -1 \\
-1 & 1 & 0 & \ldots & 0 \\
-1 & 0 & 1 & \ldots & 0 \\
\vdots & \vdots & \vdots & \ddots & \vdots \\
-1 & 0 & 0 & \ldots & 1
\end{bmatrix}.\label{eq:Star01}
\end{align}
and has eigenvalues,
\begin{align}
\lambda_{n} = \left\{ 
\begin{array}{rl}
0,&n=1\\
1,&n\in\{2,\dots,N-1\}\\
N,& n=N.
\end{array}
\right.
\label{eq:star_lam}
\end{align}
The corresponding eigenvectors are given by
\begin{align}
\bm{v}^{(1)}&=\frac{1}{\sqrt{N}}[1,\dots,1]^T\nonumber\\
\bm{v}^{(N)}&=\frac{1}{\sqrt{N^2-N}} [N-1,-1,\dots,-1]^T\label{eq:star_v},
\end{align} 
and the remaining eigenvectors $\{ \bm{v}^{(n)}\}$ form an orthonormal basis for the subspace, $\mathbb{R}^N\setminus \text{span}\{\bm{v}^{(1)},\bm{v}^{(N)}\}$. In particular, they must be orthonormal and satisfy $v_1^{(n)}=0$ and $0=\sum_mv_m^{(n)}$. It follows that SAF obtains a minimum value 
\begin{align}
\min_{\|\bm{\omega}\|=1} J(\bm{\omega},L^{(star)}) = \frac{1}{N\lambda_N^2}  =  \frac{1}{N^3}
\end{align}
when $\bm{\omega} = \bm{v}^{(N)}$ and a maximum value
\begin{align}
\max_{\|\bm{\omega}\|=1} J(\bm{\omega},L^{(star)})= \frac{1}{N\lambda_2^2}  = \frac{1}{ N}
\end{align}
when $\bm{\omega} = \bm{v}^{(2)}$.

In Fig.~\ref{fig:toy}, we illustrate (a)--(c) the chain network with $N=9$ nodes and (d)--(e) star network with $N=13$ nodes. We indicate the natural frequency vector $\bm{\omega}$ by node color, and we choose  $\bm{\omega}$ to either (a),(d) maximize $R$ by setting $\bm{\omega}=\bm{v}^{(N)}$---thereby maximizing phase synchronization---or (b),(e) minimize $R$ by setting $\bm{\omega}=\bm{v}^{(2)}$.  In panel (c), we plot the eigenvectors $\{\bm{v}^{(n)}\}$ for the chain network given by Eq.~\eqref{eq:chain_v}, and we point out that expanding $\bm{\omega}$ onto the basis $\{\bm{v}^{(n)}\}$ for a chain is equivalent to a discrete cosine transform. In general, $\bm{v}^{(N)}$ and $\bm{v}^{(2)}$ can be respectively construed as high- and low-frequency eigenvectors due to their oscillatory behavior. We point out that high-frequency eigenvectors are also well known to be prone to localization onto nodes with large degree (c.f. pg.~24 of \cite{taylor2015eigenvector}), and this phenomenon can be observed to occur for the hub in the star network [e.g., see Fig.~\ref{fig:toy}(d) and Eq.~\eqref{eq:star_v}]. Because synchronization is enhanced by aligning $\bm{\omega}$ with the high-frequency vector $\bm{v}^{(N)}$, properties of $\bm{v}^{(N)}$ reveal intuitive properties that enhance synchronization. In particular, synchronization is enhanced by implementing negative correlation between the frequencies of neighboring nodes [e.g., see Fig.~\ref{fig:toy}(a)], as well as by a positive correlation between $|\omega_m|$ and node degree, $d_m$ [e.g., see Fig.~\ref{fig:toy}(d)]. We note that these two types of correlations were previously studied for synchrony optimization for random networks \cite{skardal2014optimal,skardal2016optimal}.

\subsection{Numerical Experiment: Effectiveness of Heterogeneity Alignment}\label{sec:saf_example}
The analysis presented in Sec.~\ref{sec:SAF} has been developed for the strong synchronization regime in which $r\approx R\approx1$. Importantly, as we showed in \cite{skardal2014optimal}, the SAF provides a theoretical framework to optimize phase synchronization of systems with diverse properties, including a wide range of values for the coupling strength $K$. That is, by optimizing a system for the $r\approx R\approx 1$ regime, one inherently widens the parameter space in which the $r\approx R\approx1$ approximation is valid. Moreover, we illustrated the effectiveness of this approach with networks having diverse properties including networks that are  {both small and large as well as both heterogeneous and homogeneous}. In fact, the only assumption is that the network must be connected (see \cite{skardal2016optimal} for a generalization of the SAF for directed networks).

We briefly support this approach with a numerical experiment in which we simulated Eq.~\eqref{eq:Kuramoto1} with $H(\theta)=\sin(\theta)$ for an undirected, random network with $N=500$ nodes and mean degree 4, which we generated using the Erd\H{o}s-R\'enyi model \cite{erdos1960evolution}. We enforced it to be connected by requiring that the nodes have minimum degree $d_{min}=2$. For this network, we simulated oscillators with natural frequencies $\bm{\omega}$ given by either (a) $\bm{v}^{(N)}$, the eigenvector that corresponds to the largest eigenvalue $\lambda_N$, or (b) $\bm{v}^{(2)}$, the eigenvector (i.e., Fiedler vector \cite{fiedler1973algebraic}) that corresponds to the smallest nonzero eigenvalue $\lambda_2$. As shown in \cite{skardal2014optimal} and Corollary~\ref{co:maxmin}, these choices maximize and minimize $R$, respectively. We present results for this experiment in Fig.~\ref{fig:Kuramoto_toy}, where panels (a) and (b) depict phase-locked states at $K=0.8$ for these two choices of natural frequencies. In panel (c), we depict $r$-versus-$K$ synchronization profiles for these two systems. 

\section{Perturbation Analysis of the SAF}\label{sec:Perturb}
In this section, we develop a perturbation analysis for how the SAF [see Eq.~\eqref{eq:Objective}] is affected by structural network modifications. This analysis is built upon classical matrix perturbation theory. In Sec.~\ref{sec:Classic}, we present classical results for the perturbation of simple eigenvalues and eigenvectors of a symmetric matrix. 
In Sec.~\ref{sec:General}, we analyze general perturbations in which the Laplacian matrix $L$ undergoes a symmetric perturbation. In Sec.~\ref{sec:Add}, we study the addition and removal of edges.  In Sec.~\ref{sec:accuracy1}, we support the accuracy of the first-order approximation with a numerical experiment.

\subsection{Classical Spectral Perturbation Results \cite{atkinson2008introduction}}\label{sec:Classic}
We begin by presenting a well-known result that describes the first-order perturbation of eigenvalues and eigenvectors of a symmetric matrix $L$.

\begin{theorem}[Perturbation of Simple Eigenvalues and their Eigenvectors \cite{atkinson2008introduction}] 
Let $L$ be a symmetric matrix with simple eigenvalues $\{\lambda_n\}$ and normalized eigenvectors $\{\bm v^{(n)}\}$.
Consider a fixed symmetric perturbation matrix $\Delta L$, and let $L(\epsilon)=L+\epsilon \Delta L$.
Denote the eigenvalues  {and eigenvectors} of $L(\epsilon)$ by $\lambda_n(\epsilon)$  {and $\bm{v}^{(n)}(\epsilon)$, respectively, for $n=1,2,\dots,N$.
It follows that
\begin{align}
\lambda_n(\epsilon)  &= \lambda_n + \epsilon \lambda'(0) + \mathcal{O}(\epsilon^2),\nonumber \\
{\bm v^{(n)}} (\epsilon) &= \bm v^{(n)} + \epsilon {\bm{v}^{(n)}}'(0) + \mathcal{O}(\epsilon^2),
\label{eq:First2a}
\end{align}
and the derivatives with respect to $\epsilon$ at $\epsilon=0$ are given by
\begin{align}
\lambda_n'(0) & = (\bm v^{(n)})^T \Delta L \bm v^{(n)}\nonumber \\
 {\bm{v}^{(n)}}'(0) &= \sum_{{m\not=n}} \frac{ (\bm v^{(m)})^T \Delta L \bm v^{(n)}}{\lambda_n-\lambda_m} \bm{v}^{(m)}.
\label{eq:First2}
\end{align}}
\end{theorem}

\begin{proof}
See \cite{atkinson2008introduction}.
\end{proof}
 {\begin{remark} 
Note for $n=1$ that $\lambda_1(\epsilon)=0$ and $\bm{v}^{(1)}(\epsilon)= N^{-1/2}\bm{1}$ for any $\epsilon$ since the perturbation $\Delta L$ has the same null space as $L$, which is $\text{span}(\bm{1}) $.
\end{remark} }

Due to continuity, the approximations in Eq.~\eqref{eq:First2a} are accurate when the perturbations are small. However, the regime for which such approximation is valid (i.e., how small $\epsilon$ needs to be) generally depends on  {$L$, $\epsilon$, and} the perturbation matrix $\Delta L$.

\subsection{General Network Perturbations}\label{sec:General}
We now present a first-order expansion of the SAF that is analogous to the expansions given by Eq.~\eqref{eq:First2a}.

\begin{theorem}[Perturbation of the SAF under a Network Modification]
Let $J(\bm{\omega},L)$ denote the SAF given by Eq.~\eqref{eq:Objective} for natural frequencies $\bm{\omega}$ and symmetric network Laplacian $L$, and let $J(\bm{\omega},L(\epsilon))$ denote the SAF for the network after it undergoes a  {symmetric} modification $\epsilon \Delta L$. Assume the eigenvalues of $L$ and $L(\epsilon)=L+\epsilon \Delta L$ are simple, and that the original and perturbed networks are both connected. Then the first-order expansion in $\epsilon$ for  {the perturbed SAF is given by
\begin{align}
J(\bm{\omega},L(\epsilon)) = J(\bm{\omega},L) + \epsilon J'(\epsilon) +\mathcal{O}(\epsilon^2),\label{eq:First_J}
\end{align}
where
\begin{align}
J'(\epsilon) &=     \frac{2}{N}\sum_{n=2}^N   \left(\frac{ \bm{\omega}^T  \bm v^{(n)} }{   \lambda_n^3}  \right)
\left( {\sum_{m=2}^{N}}
\frac{ [ \bm{\omega}^T\bm v^{(m)}][(\bm v^{(m)})^T \Delta L \bm v^{(n)}]}{(1-\lambda_m/\lambda_n)-\delta_{nm}}   \right) .
\end{align}
}

\begin{proof}
See Appendix~\ref{appendixC}.
\end{proof}
\end{theorem}

\begin{remark}\label{remark:rem1}
Due to continuity, Eq.~\eqref{eq:First_J} is accurate when the perturbation is small, i.e., $|\Delta J| \ll J$. Because Eq.~\eqref{eq:First_J} relies on Eq.~\eqref{sec:Classic}, one heuristic to ensure accuracy is that we require Eq.~\eqref{sec:Classic} to be accurate for every eigenvalue, which is expected when $\epsilon (\bm v^{(n)})^T \Delta L \bm v^{(n)}/\lambda_n \ll1$ for every $n=2,3,\dots,N$. (Recall that $\lambda_1$ is always zero.) This suggests $\epsilon/\lambda_2\ll1$, and we provide numerical support for this heuristic in Sec.~\ref{sec:accuracy1}. 
However, we conjecture that this heuristic may be too strong (i.e., sufficient but not necessary). We  {consider} $\epsilon/\overline{\lambda}\ll1$ to be a reasonable heuristic in many situations, where $\overline{\lambda}=N^{-1}\sum_{n} \lambda_n$. 
\end{remark}
\begin{remark}
The computation of Eq.~\eqref{eq:First_J} requires  {$\mathcal{O}(MN+N^2)$}
multiplications, where $M$ is the number of nonzero entries in $\Delta L$. In contrast, direct computation of the new SAF requires solving $N-1$ eigenvalues and eigenvectors, which  {typically} involves $\mathcal{O}(N^3)$ multiplications  {in practice}, and computing Eq.~\eqref{eq:Objective} involves  {$\mathcal{O}(N^2)$}
multiplications. Therefore, for large networks and sparse $\Delta L$ (i.e., $M\ll  {\mathcal{O}(N^2)}$), the perturbation result is much more efficient to compute, and in particular, it is $\mathcal{O}(N^2)$ versus $\mathcal{O}(N^3)$.
\end{remark}

\subsection{Edge Additions and Removals}\label{sec:Add}
Equation~\eqref{eq:First_J} gives a first-order approximation to the change in the SAF due to any symmetric perturbation $\epsilon\Delta L$ of the Laplacian $L$. We now provide a more specific result for the addition and removal of undirected, unweighted edges. 

\begin{corollary}[Perturbation of the SAF under Edge Modifications]\label{Corr_11}
Consider the SAF given by Eq.~\eqref{eq:Objective} and the  {perturbation} of undirected edge $(p,q)$  {(e.g., $A_{pq} \mapsto A_{pq}  \pm \epsilon$ and $A_{pq} \mapsto A_{pq}  \pm \epsilon$)} and define 
\begin{align}
Q_{pq}&=
 \frac{2}{N}\sum_{n=2}^N   \left(\frac{ \bm{\omega}^T  \bm v^{(n)} }{   {\lambda_n ^3}}  \right)
\left(\sum_{{m=1}}^N \frac{ [ \bm{\omega}^T\bm v^{(m)}][(\bm v^{(m)}_p-\bm v^{(m)}_q)(\bm v^{(n)}_p-\bm v^{(n)}_q)]}{(1{ -\lambda_m/\lambda_n) -\delta_{nm}} }   \right) ,
 \label{eq:pert_1edge}
\end{align}
then Eq.~\eqref{eq:First_J} has the simplified form
\begin{align}
 {J(\bm{\omega},L(\epsilon)) = J(\bm{\omega},L)}  \pm \epsilon Q_{pq} +\mathcal{O}(\epsilon^2), \label{eq:dJ_1edges}
\end{align}
where $+$ and $-$ correspond to edge addition and subtraction, respectively.

 \begin{proof}
 See Appendix \ref{appendixD}.
 \end{proof}
\end{corollary}

\begin{corollary}[Perturbation of the SAF under Subgraph Rewiring]\label{Corr_12}
Consider the SAF given by Eq.~\eqref{eq:Objective} and a network in which a set of edges $\mathcal{E}^{(+)}\subseteq\{1,\dots,N\}\times \{1,\dots,N\}$ are added and a set of edges $\mathcal{E}^{(-)}\subseteq\{1,\dots,N\}\times \{1,\dots,N\}$ are removed, then Eq.~\eqref{eq:First_J} has the simplified form
\begin{align}
 {J(\bm{\omega},L(\epsilon)) = J(\bm{\omega},L)} +  \sum_{(p,q)\in\mathcal{E}^{(+)}}  \epsilon Q_{pq} -\sum_{(p,q)\in\mathcal{E}^{(-)}}  \epsilon Q_{pq} +\mathcal{O}(\epsilon^2) . \label{eq:dJ_2edges}
\end{align}
\begin{proof}
See Appendix \ref{appendixE}
\end{proof}
\end{corollary}

\subsection{Numerical Experiment:  {Validation} of the First-Order Approximation}\label{sec:accuracy1}

We  {now} present a numerical experiment to illustrate the accuracy of Eq.~\eqref{eq:First_J} and Eq.~\eqref{eq:dJ_1edges} by comparing predicted and observed values of the SAF upon edge additions. In particular, we considered a system given by Eq.~\eqref{eq:HLD} in which the natural frequencies $\{\omega_n\}$ were randomly drawn from a normal distribution, and we constructed undirected, scale-free networks using the configuration model \cite{bekessy1972asymptotic}. We generated networks with degrees $\{d_i\}$ following the distribution $P(d)\propto d^{-\gamma}$ with $\gamma=2.5$, and either (a) $N=100$ and $d_{{min}}=5$ or (b) $N=250$ and $d_{{min}}=25$. We considered single-edge additions for each system, and for each new edge $(p,q)$, we compared the observed change to the SAF,  {$\Delta J=J(\bm{\omega},L(\epsilon)) - J(\bm{\omega},L)$}, and the first-order approximation $Q_{pq}$ given by Eqs.~\eqref{eq:pert_1edge}  {and~\eqref{eq:dJ_1edges}}.

\begin{figure*}[t]
\centering
\includegraphics[width=1\linewidth]{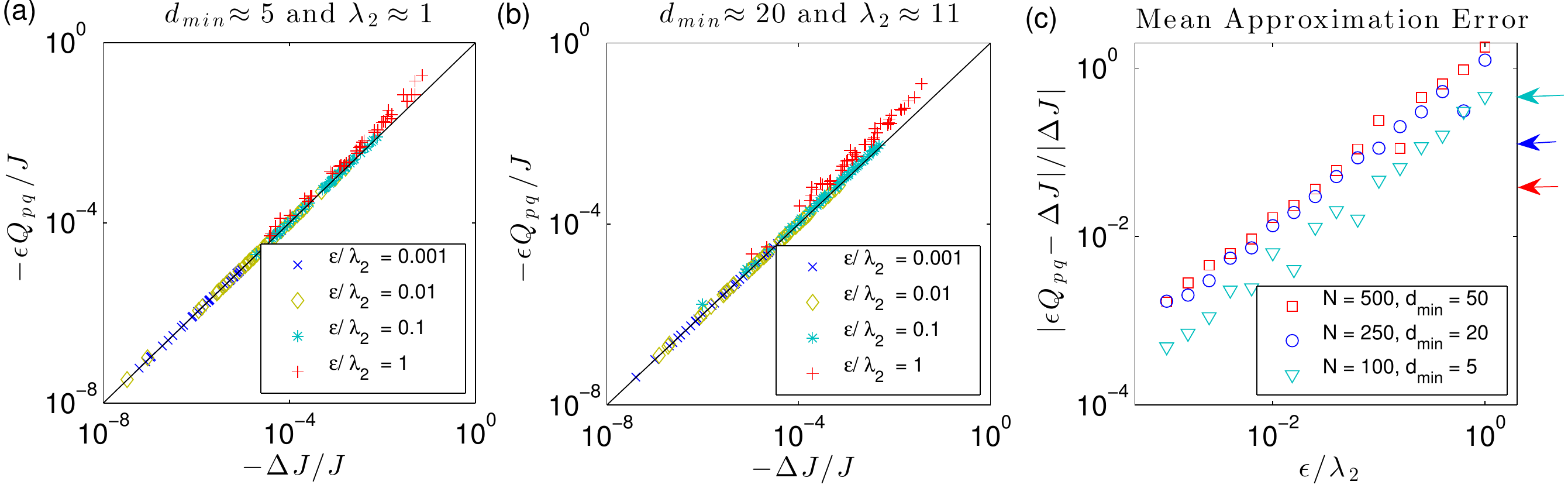}
\caption{ 
Approximation accuracy of Eq.~\eqref{eq:dJ_1edges} for the addition of 50 randomly selected edges. (a),(b) Scatter plots of the first-order prediction $ Q_{pq}$ versus actual change $\Delta J$ to SAF after we add an edge to scale-free networks, which we constructed using the configuration model \cite{bekessy1972asymptotic} with exponent $\gamma=2.5$ and either (a) $N=100$ and $d_{{min}}=5$ or  (b) $N=250$ and $d_{{min}}=25$. By varying $\epsilon$, we show results for several choices of $\epsilon/{\lambda}_2$. 
(c)~We plot the mean approximation error versus $\epsilon/\lambda_2$ for networks of different size $N$ and minimum degree $d_{min}$. Results indicate the mean across 50 randomly selected edge additions. The arrows indicate the error when $\epsilon=1$, which vanishes with growing $\lambda_2$.
}
\label{fig:accuracy1}
\end{figure*}

We plot these results in Fig.~\ref{fig:accuracy1}, and we describe the perturbation size in terms of the ratio $\epsilon/{\lambda}_2$ (see Remark \ref{remark:rem1}). In panels (a) and (b), we plot predicted versus true values of $\Delta J$ for various values of $\epsilon$ for two scale-free networks. Results indicate 50 randomly selected edge additions. In panel (c), we plot the mean approximation error---that is, the mean  {fractional error, $|\epsilon Q_{pq}-\Delta J|/|\Delta J|$}, across 50 edge additions---as a function of $\epsilon/\lambda_2$, for several networks of different size and minimum degree. The arrows indicate the approximation error when $\epsilon=1$ (i.e., the addition of an undirected edge). Our first observation is that the approximation error vanishes with growing network size and density (i.e., increasing $d_{min}$). For example, the mean error is approximately 2\% for the network with $N=500$ and $d_{min}=50$, whereas it is approximately 40\% for the network with $N=100$ and $d_{min}=5$. Our second observation is that even when the mean approximation error is somewhat large (e.g., 40\%), Eq.~\eqref{eq:pert_1edge} still captures the correct magnitude of the perturbation  {of $J$}, and this is significant because $\Delta J$ can vary by several orders of magnitude for the different edge perturbations [see panels (a) and (b)].

\section{ {Ranking Edges via Perturbation to the SAF}}\label{sec:OptimalSection}

 {In this section}, we use our perturbation analysis as a centrality measure \cite{taylor2015eigenvector} to rank the edges and potential edges according to their importance to the SAF.  {This ranking is akin to other rankings that are specific to a particular class of dynamics, including PageRank (which is important to random walks \cite{gleich2015pagerank} and collective behavior \cite{skardal2015collective}) and dynamical importance \cite{restrepo2006characterizing} (which is important to dynamics ranging from epidemic spreading to synchronization). For the ranking that we introduce here, the top-ranked edge is the one that yields the minimal SAF, and therefore maximal $R$} upon its removal. Similarly, the top-ranked potential edge is  {one that yields the minimal SAF, and therefore maximal $R$} upon its addition. Importantly, this approach takes takes into account  {both the structure \emph{and} dynamics of the} system---that is, both the particular network structure and the oscillators' heterogeneous natural frequencies.

This section is organized as follows: 
In Sec.~\ref{sec:ranking}, we rank the edges according to their importance to the SAF (and thus phase synchronization). 
In Sec.~\ref{sec:Maximize}, we define a class of optimization problem that maximizes phase synchronization with edge modifications. 
In Sec.~\ref{sec:toy_2}, we identify the top-ranked potential edges that can be added to the pedagogical chain network.

\subsection{Ranking Edges According to the SAF}\label{sec:ranking}

We first introduce some notation. Let $G(\mathcal{V},\mathcal{E})$ define a network with a set of nodes $\mathcal{V}=\{1,\dots,N\}$ and a set of {undirected} edges, $\mathcal{E}\subseteq\mathcal{V}\times \mathcal{V}$. We disallow self-edges so that $\{ (n,n)\}\cap\mathcal{E}=\emptyset$. For a given set of edges $\mathcal{E}$, we define a set of complementary edges (i.e., potential edges) 
$\mathcal{PE}=\mathcal{V}\times\mathcal{V}\setminus\left(\mathcal{E}\cup \{ (n,n)\} \right)$. The sets $\mathcal{E}$ and $\mathcal{PE}$ define the edges that can  {potentially} be removed and added, respectively.

 {We now introduce the rankings.}

\begin{definition}[SAF-Induced Ranking of Edges]
Given a connected network $ {G}=(\mathcal{V},\mathcal{E})$ with symmetric Laplacian matrix $L$ and a frequency vector $\bm{ \omega}$, we rank each edge $(p,q)\in\mathcal{E}$ according to the first-order approximation for the perturbation of the SAF that is induced by its removal, $\Delta J\approx  {-Q_{pq}}$. 
Specifically, we define
\begin{align}
X(p,q) = 1+ | \mathcal{E}'|,~\text{where}~ \mathcal{E}' = \{ (n,m)\in\mathcal{E} : Q_{nm} > Q_{pq}\}\label{eq:rank1}
\end{align}
so that $X(p,q)\in\{1,\dots,|\mathcal{E}|\}$ defines the rank of each edge $(p,q)\in\mathcal{E}$. 
\end{definition}

\begin{definition}[SAF-Induced Ranking of Potential Edges]
Given a connected network $ {G}=(\mathcal{V},\mathcal{E})$ with symmetric Laplacian matrix $L$ and a frequency vector $\bm{\omega}$, we rank each potential edge $(i,j)\in\mathcal{PE}$ according to the first-order approximation for the perturbation of the SAF that is induced by its addition, $\Delta J\approx Q_{pq}$. We define
\begin{align}
Y(p,q) = 1+  | \mathcal{PE}'|,~\text{where}~ \mathcal{PE}' = \{ (n,m)\in\mathcal{PE} : Q_{nm} <Q_{pq}\}\label{eq:rank2}
\end{align}
so that $Y(p,q)\in\{1,\dots,|\mathcal{PE}|\}$ defines the rank of each potential edge $(p,q)\in\mathcal{PE}$.
\end{definition}

We note that it is generally possible for more than one edge correspond to a given value $Q_{nm}$, and in this situation the rankings $\{X(n,m)\}$ of edges $\mathcal{E}$ and $\{Y(n,m)\} $ of potential edges $\mathcal{PE}$ can lead to ties. That is, multiple edges will have an identical rank, and the next-ranked edge will have a rank that takes into account the number of edges that are tied. For some applications (e.g., the algorithms we develop in the following section), it can be necessary that there are no ties, and in this case we break the tie by randomly assigning an appropriate rank to the edges that correspond to an identical $Q_{nm}$ value.

\subsection{Optimizing Phase Synchrony with Edge Modifications}\label{sec:Maximize}

We will use the rankings $\{X(n,m)\}$ and $\mathcal{E}$ and $\{Y(n,m)\} $ to efficiently solve the following optimization problem. 

\begin{definition}[Maximal Phase Synchrony with Edge Modifications]
Let $R(\bm{\omega},L)$ denote the variance order parameter given by Eq.~\eqref{eq:rObjective} of the phase locked solution of Eq.~\eqref{eq:HLD} for natural frequencies $\bm\omega$ and  {network} Laplacian  $L$. Through the removal of $T^{(-)}$ edges and the addition of $T^{(+)}$ new edges, we wish to solve
\begin{align}
\max_{\Delta L \in \mathcal{D}^{(T^{(-)},T^{(+)})} } R(\bm{\omega},L+\Delta L), 
\label{def:main_optimization_problem}
\end{align}
where 
\begin{align}
\mathcal{D}^{(T^{(-)},T^{(+)})} = \left\{\Delta L : \Delta L  = \sum_{(p,q)\in \mathcal{E}^{(+)}}  \Delta L^{(pq)} -  \sum_{(p,q)\in \mathcal{E}^{(-)}} \Delta L^{(pq)}  \right\} 
\label{def:perturbation_ensemble}
\end{align}
is the ensemble of appropriate perturbations to the Laplacian $L$ that can be obtained by removing $T^{(-)}$ edges, $\mathcal{E}^{(-)}\subseteq \mathcal{E}$, and adding $T^{(+)}$ new edges, $\mathcal{E}^{(+)}\subseteq  {\mathcal{PE}}$, and 
\begin{align}
\Delta L_{ij}^{(pq)} = \left\{ \begin{array}{rl} 
1,&  {(i,j)\in\{(p,p),(q,q)\}}\\ 
-1,&  {(i,j)\in\{(p,q),(q,p)\}}\\ 
0,&\text{otherwise}.
\end{array}\right. 
\label{eq:DL}
\end{align}
gives the change in $L$ due to the addition of an edge $(p,q)$.
\end{definition}

Because $R$ can be solved in terms of the SAF for HLD system [see Eq.~\eqref{eq:rObjective}], Eq.~\eqref{def:main_optimization_problem} is equivalent to
\begin{align}
\min_{\Delta L \in  {\mathcal{D}^{(T^{(-)},T^{(+)})}}} J(\bm{{\omega}},L+\Delta L) .
\label{def:SAF_optimization_problem}
\end{align}
Both Eq.~\eqref{def:main_optimization_problem} and Eq.~\eqref{def:SAF_optimization_problem} can be solved with an exhaustive search if $N$, $T^{(-)}$ and $T^{(+)}$ are very small. However, this approach is infeasible for practical situations in which the network is large or  {more than a few} edges are modified, and one must instead search for approximate solutions that  {can be computed efficiently.} 

\begin{figure}[t]
\centering
\includegraphics[width=1\linewidth]{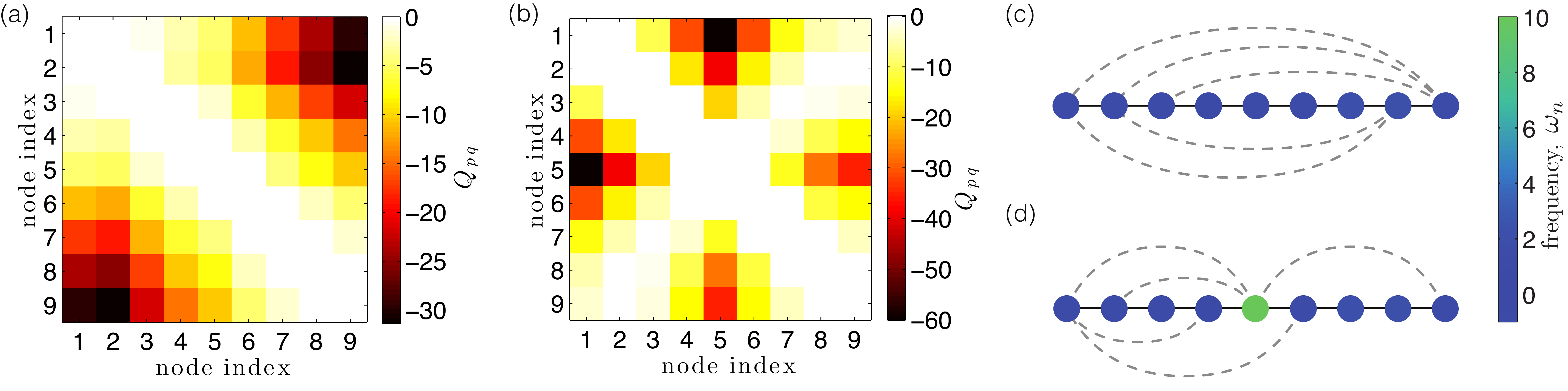}
\caption{
Perturbation $Q_{pq}$ given by Eq.~\eqref{eq:pert_1edge} with $\epsilon=1$ for potential edges $(p,q)\in\mathcal{PE}$ for the chain network with $N=9$ nodes and two choices for $\bm{\omega}$: (a) $\{\omega_n\}$ are independently drawn from a normal distribution with unit variance, and (b) $\{\omega_n\}$ are the same as those in (a) except we create an outlier oscillator by setting $\omega_5=10$. We indicate by dashed curves in panels (c) and (d), respectively, the five top-ranked potential edges, $Y(p,q)\in\{1,\dots,5\}$ given by Eq.~\eqref{eq:rank2}, for the $Q_{pq}$ values shown in panels (a) and (b). Node color indicates $\omega_n$.
}
\label{fig:toy_2}
\end{figure}

\subsection{SAF-Based Edge Ranking for Chain Network}\label{sec:toy_2}

Before continuing, we present a numerical experiment to highlight that the rankings introduced in Sec.~\ref{sec:ranking} take into account both the network structure and oscillator dynamics (i.e., their natural frequencies $\{\omega_n\}$). That is, depending on the particular system it is possible for the rankings to be dominated by either the network structure or natural frequencies. We illustrate this phenomenon by studying the ranking of potential new edges for the chain network that was described in Sec.~\ref{sec:toy_examples} as a pedagogical network for the SAF. In this study, we computed $Q_{pq}$ for all possible edge additions $(p,q)\in\mathcal{PE}$ for two choices of natural frequencies: (a) $\{\omega_n\}$ are drawn independently from a normal distribution with unit variance, and (b) $\{\omega_n\}$ are identical to those in (a) except we define $\omega_5=10$ for oscillator $n=5$. The motivation for setting $\omega_5=10$ is that this oscillator becomes an \emph{outlier} in that its natural frequency is much larger than any other oscillator (i.e., its magnitude is 10 times larger than the standard deviation of the other oscillators). 

In Fig.~\ref{fig:toy_2}(a) and \ref{fig:toy_2}(a), we depict the values $\{Q_{pq}\}$ for these two choices for $\bm{\omega}$. In panels (c) and (d), respectively, we indicate by dashed curves the edges that correspond to the five top-ranked potential edges, $Y(p,q)\in\{1,\dots,5\}$ given by Eq.~\eqref{eq:rank2}, for the $Q_{pq}$ values shown in panels (a) and (b). Note in panel (c) that the top-ranked potential edges connect together the ends of chain, which significantly changes the topology of the network and can be measured, for example, via the network diameter (which decreases from 8 to 4). In contrast, in the presence of the outlier oscillator, node $n=5$, the top-rank edges connect to the outlier or its neighbors to mitigate its disruptive effect on synchronization. In the following section, we present formal algorithms that use the rankings of edges and potential edges to solve the optimization problem defined in Sec.~\ref{sec:Maximize}.

\section{ {Gradient-Descent Algorithms for Synchrony Optimization}}\label{sec:algorithm_sec}
In \cite{skardal2014optimal}, we developed accept/reject (i.e., Monte Carlo) rewiring algorithms to approximately minimize the SAF---thereby maximizing phase synchronization. That is, we developed a process in which we iteratively  {propose} an edge rewire (which we selected uniformly at random),  {compute} the new SAF after the rewire, and then  {accept or reject} the proposed rewiring based on whether or not the SAF  {decreases}. Although we showed that this approach is effective for optimizing the synchronization properties of  {several types of} networks, it is important to develop more efficient algorithms to address practical applications. We now leverage the results of  {Secs.~\ref{sec:Perturb} and \ref{sec:OptimalSection}} to develop gradient-descent algorithms that efficiently identify network modifications that optimally enhance phase synchronization.

This section is organized as follows: 
In Sec.~\ref{sec:Algorithms}, we develop gradient-descent algorithms based on the rankings to efficiently solve these optimization problems. 
In Sec.~\ref{sec:OptimizeNumerics}, we support these results with numerical experiments.
In Sec.~\ref{sec:unreachable}, we provide an extended study of synchrony optimization under non-ideal scenarios.

\subsection{Gradient-Descent Algorithms}\label{sec:Algorithms}

We now describe two algorithms that can be used to approximately solve the class of optimization problem  {defined} in Sec.~\ref{sec:Maximize}. The first algorithm is formally presented in Algorithm~\ref{alg:A}, which we now describe. It consists of two steps. First, we remove the $T^{(-)}$ edges that have lowest rank, $\mathcal{E}^{(-)} = \{(n,m)\in\mathcal{E}:X(n,m)\ge |\mathcal{E}|- T^{(-)}\}$. Next, we add the $T^{(+)}$ potential edges that have highest rank, $\mathcal{E}^{(+)} = \{(n,m)\in\mathcal{PE}:Y(n,m)\le T^{(+)}\}$. To implement this algorithm, we assume there are no tied rankings so that $X(n,m)\not=X(p,q)$ and $Y(n,m)\not=Y(p,q)$ whenever $(n,m)\not=(p,q)$.

We note that Algorithm~\ref{alg:A} is a 1-step gradient descent algorithm since the gradient of the SAF (i.e., its first-order approximation) due to the subgraph rewiring is given by Eq.~\eqref{eq:dJ_2edges}. In particular, the selections of edges $\mathcal{E}^{(+)}$ and $\mathcal{E}^{(-)}$ according to Algorithm~\ref{alg:A} correspond to the direction of the largest gradient. Also, due to Eq.~\eqref{eq:rObjective}, the gradient of the SAF equals the negative gradient of $R$ for the phase-locked state of the system given by Eq.~\eqref{eq:HLD}.
However, we also note that Eq.~\eqref{eq:dJ_2edges} is an approximation to the actual change $\Delta J$ that will occur to the SAF, and therefore Algorithm~\ref{alg:A} only approximately solves the class of optimization problem given by Eq.~\eqref{def:main_optimization_problem}. In fact, the solution error grows with the error of the first-order approximation (see Remark \ref{remark:rem1}). Importantly, the accuracy of Eq.~\eqref{eq:dJ_2edges} decreases with increasing number of edge manipulations, $|\mathcal{E}^{(-)}|+|\mathcal{E}^{(+)}|$, and therefore we expect the performance of Algorithm~\ref{alg:A} to  {become worse} as this number increases. To obtain better approximate solutions to the optimization problem given by Eq.~\eqref{def:main_optimization_problem} with large $T^{(-)}$ or  $T^{(+)}$, we now introduce a second algorithm.

\begin{algorithm}[t]
\caption{{\bf Rank-Based Modifications without Updating }}
\begin{algorithmic}[1]\label{alg:A}
\REQUIRE Network with edges $\mathcal{E}$, potential edges $\mathcal{PE}$, natural frequency vector $\bm{\omega}$, and numbers of edge additions, $T^{(+)}$, and removals, $T^{(-)}$
\ENSURE Set of edges to be added, $\mathcal{E}^{(+)}$, and removed, $\mathcal{E}^{(-)}$ 
\STATE Rank edges $ \mathcal{E}$ and potential edges $ \mathcal{PE}$ according to Eqs.~\eqref{eq:rank1} and \eqref{eq:rank2}
\STATE Define $\mathcal{E}^{(+)}$ as the top-ranked edges, $\mathcal{E}^{(+)} = \{(p,q) : X_{pq} \ge |\mathcal{E}|-T^{(+)} \}$
\STATE Define $\mathcal{E}^{(-)}$ as the lowest-ranked edges, $\mathcal{E}^{(-)} = \{(p,q) : Y_{pq} \le T^{(-)} \}$
\end{algorithmic}
\end{algorithm}

\begin{algorithm}[t]
\caption{{\bf Rank-Based Modifications with Updating}}
\label{alg:B}
\begin{algorithmic}[1]
\REQUIRE Network with edges $\mathcal{E}$, potential edges $\mathcal{PE}$, natural frequency vector $\bm{\omega}$, and numbers of edge additions, $T^{(+)}$, and removals, $T^{(-)}$
\ENSURE Set of edges to be added, $\mathcal{E}^{(+)}$, and removed, $\mathcal{E}^{(-)}$ 
\STATE Initialize sets of edges, $\hat{\mathcal{E}}=\mathcal{E}$, and potential edges, $\hat{\mathcal{PE}}=\mathcal{PE}$
\STATE Initialize the sets of edges to be added, ${\mathcal{E}}^{(+)}=\emptyset$, and removed, ${\mathcal{E}}^{(-)}=\emptyset$
\FOR{ $t\in\{1,\dots,\max(T^{(-)},T^{(+)})\}$} 
\IF{$t\le T^{(-)}$}
\STATE Identify lowest-ranked edge $(p^*,q^*)\in\hat{\mathcal{E}}$ such that $X_{pq} = |\hat{\mathcal{E}}|$
\STATE Add lowest-ranked edge to removal set, $\mathcal{E}^{(-)} = \mathcal{E}^{(-)} \cup \{(p^*,q^*) \}$
\STATE Update the set of edges $\hat{\mathcal{E}} = \hat{\mathcal{E}} \setminus \{(p^*,q^*)\}$
\ENDIF
\IF{$t\le T^{(+)}$}
\STATE Identify top-ranked potential edge $(p^*,q^*)\in\hat{\mathcal{PE}}$ such that $Y_{pq} = 1$
\STATE Add top-ranked potential edge to addition set, $\mathcal{E}^{(+)} = \mathcal{E}^{(+)} \cup \{(p^*,q^*) \}$
\STATE Update the set of potential edges $\hat{\mathcal{PE}} = \hat{\mathcal{PE}} \setminus \{(p^*,q^*)\}$
\ENDIF
\ENDFOR
\end{algorithmic}
\end{algorithm}

We present in Algorithm~\ref{alg:B} another algorithm that utilities the rankings of edges and potential edges according to the SAF. The main difference from Algorithm~\ref{alg:A} is that in Algorithm~\ref{alg:B}, the edge modifications are made  {\it sequentially} rather than simultaneously.  {That is, after each edge modification, the eigenvalues and eigenvectors of the resulting network Laplacian matrix are computed.} In this way, it is a multi-step gradient-descent algorithm. In particular, we first remove the lowest-ranked edge and add the top-ranked potential edge. Then we compute the new rankings after the edge rewire. Next, according to these new rankings, we again remove the lowest-ranked edge, add the top-ranked potential edge, and compute the new rankings. We repeat this process until $T^{(-)}$ edges are removed and $T^{(+)}$ edges are added.

The main benefit of Algorithm~\ref{alg:B} is that the error of the first-order approximation for subgraph rewiring [see Eq.~\eqref{eq:dJ_2edges}] remains small by keeping the perturbations small (i.e., only one rewire is made at a time). We note that it is also possible to update the rankings between the step of edge removal and edge addition to make the perturbations even smaller, but we do not explore this option. We find that Algorithm~\ref{alg:B} yields improved approximate solutions for the optimization problem given by Eq.\eqref{def:main_optimization_problem}; however, it does so at an increased computational cost. In particular, whereas the matrix $\{Q_{pq}\}$ is calculated only once for Algorithm~\ref{alg:A}, it must be recalculated after each of the rewires for Algorithm~\ref{alg:B}. For some applications, we expect that will be beneficial to modify Algorithm~\ref{alg:B} so that the matrix $\{Q_{pq}\}$ is updated after a few (and not every) rewire, and we leave this direction open for future work. Moreover, Algorithm~\ref{alg:B} implements a 1-to-1 modification strategy in which we remove an edge, add an edge, and repeat; however, one could also explore different strategies for the ordering in which edges are removed and added (e.g., one could first removal all edges $\mathcal{E}^{(-)}$ and then add the new edges $\mathcal{E}^{(+)}$, or vice versa). Therefore, although we focus on two algorithms, we stress that the results presented in Secs.~\ref{sec:SAF} and \ref{sec:Perturb} provide a mathematical foundation that can serve as a starting point for developing even further optimization algorithms for phase synchronization in oscillator networks.

\begin{figure*}[t!]
\centering
\includegraphics[width=1\linewidth]{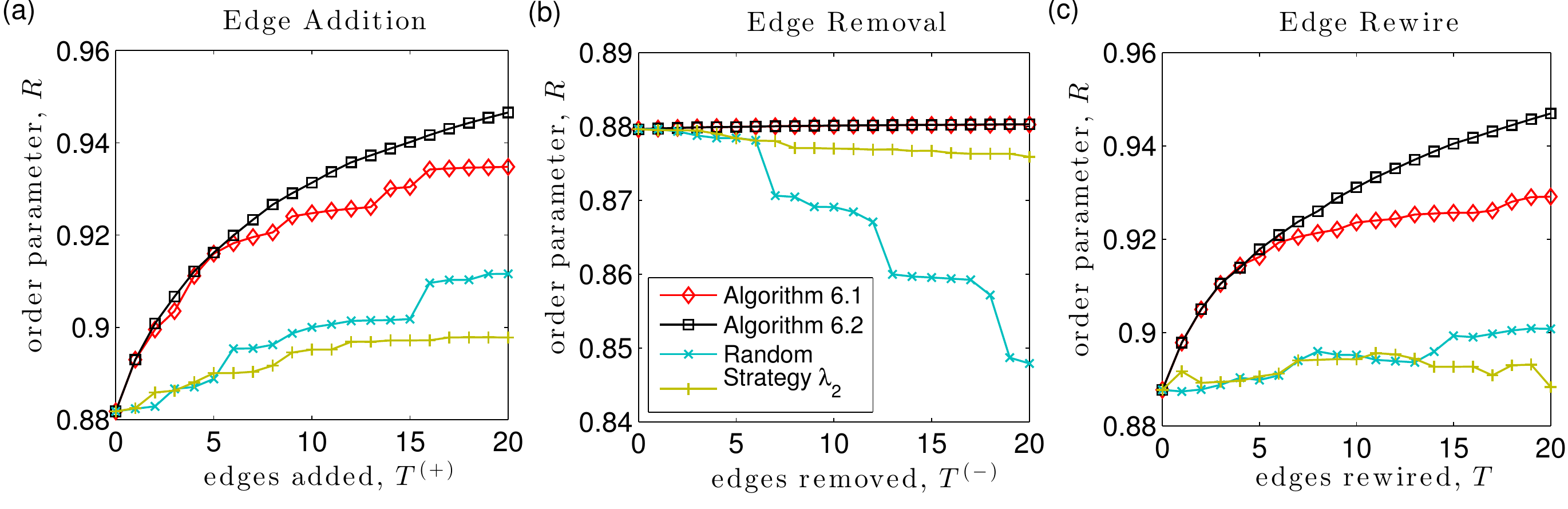}
\caption{
Maximizing phase synchronization with optimal edge modifications. In panels (a), (b), and (c), we illustrate the effectiveness of Algorithms~\ref{alg:A} and \ref{alg:B} for the class of optimization problem defined in Eq.~\eqref{def:main_optimization_problem}. In particular, we study (a) edge addition by setting $T^{(-)}=0$ and allowing $T^{(+)}$ to vary, (b) edge removal by setting $T^{(+)}=0$ and allowing $T^{(-)}$ to vary, (c) edge rewiring by setting $T^{(-)}=T^{(+)}=T$ and allowing $T$ to vary. We compare Algorithms~\ref{alg:A} and \ref{alg:B} to two other edge modification algorithms: Strategy ``Random'' corresponds to when the edges are added or removed uniformly at random, and ``Strategy $\lambda_2$'' corresponds to when the edges are added or removed so as to maximize eigenvalue $\lambda_2$, which is the network's algebraic connectivity \cite{fiedler1973algebraic}. In all panels, the initial network is scale-free with $N=50$ nodes, exponent $\gamma=2.5$, and minimum degree $d_\textrm{min}=10$. The coupling strength is $K=0.02$  {and the values of $R$ are given by Eq.~\eqref{eq:rObjective}.}
}
\label{fig:compare}
\end{figure*}

\subsection{Numerical Experiment: Enhancing Phase Synchronization with Edge Modifications}
\label{sec:OptimizeNumerics}

We now support Algorithms~\ref{alg:A} and \ref{alg:B} with numerical experiments. We constructed an initial system given by Eq.~\eqref{eq:HLD} with natural frequencies $\{\omega_n\}$ drawn from a normal distribution, and we randomly assigned them to nodes in a scale-free network with $N=50$ nodes, exponent $\gamma=2.5$, and minimum degree $d_{{min}}=10$, which we constructed using the configuration model \cite{bekessy1972asymptotic}.
We conducted three experiments for the class of optimization problem defined by Eq.~\eqref{def:main_optimization_problem}: 
\begin{itemize}
\item[(a)] We studied the effect of edge additions and no edge removals by setting $T^{(-)}=0$ and considering various $T^{(+)}$.
\item[(b)] We studied the effect of edge removals and no edge additions by setting $T^{(+)}=0$ and considering various $T^{(-)}$.
\item[(c)] We studied the effect of rewiring $T$ edges by setting $T^{(-)}=T^{(+)}=T$ and considering various $T$.
\end{itemize}

In Fig.~\ref{fig:compare}(a), (b) and (c), we plot the linear order parameter $R$  {given by Eq.~\eqref{eq:rObjective}} versus $T^{(+)}$, $T^{(-)}$ and $T$ for the solutions that were obtained by Algorithms~\ref{alg:A} and \ref{alg:B} for these respective optimization problems. Note that Algorithm~\ref{alg:B} provides better solutions than Algorithm~\ref{alg:A}; however, Algorithm~\ref{alg:A} performs nearly as good when the number of modifications is small. Interestingly, we find that depending on the edge choice, both edge addition and removal can possibly increase or decrease $R$. By comparing panel (b) to (a), however, one can observe  for this experiment that edge addition is much more effective than edge removal for  {the increase of} $R$. Therefore, the enhanced synchronization that can be observed in panel (c) is mostly due to the edges that were added  {rather than the edges that were removed.}

To gauge the effectiveness of Algorithms~\ref{alg:A} and \ref{alg:B} for enhancing phase synchronization, we compare them to two other strategies for modifying a network. First, we define the ``Random'' strategy to indicate the situation in which the appropriate number of edges are removed and/or added uniformly at random. Second, we define ``Strategy $\lambda_2$'' to indicate the selection of edges so as to maximize the eigenvalue $\lambda_2$  {per step}, which is often referred to as the network's algebraic connectivity \cite{fiedler1973algebraic}.  The motivation for comparing to this approach is that $\lambda_2$ is often tuned to control the synchronization of network-coupled dynamical systems with identical oscillators \cite{barahona2002synchronization,simonsen2008transient,lozano2012role,motter2013spontaneous,nishikawa2015comparative}.
To efficiently implement Strategy $\lambda_2$, we use the first-order approximation for the perturbation of $\lambda_2$ due to a network modification as given by Eq.~\eqref{eq:First2} with $n=2$ and $\Delta L=\Delta L^{(pq)}$ given by Eq.~\eqref{eq:DL}.
Note that Algorithms~\ref{alg:A} and \ref{alg:B} both significantly outperform these baseline strategies, which do not take into account the heterogeneous dynamics (i.e., natural frequencies $\{\omega_n\}$) of the network-coupled dynamical system.

\begin{figure*}[t!]
\centering
\includegraphics[width=1\linewidth]{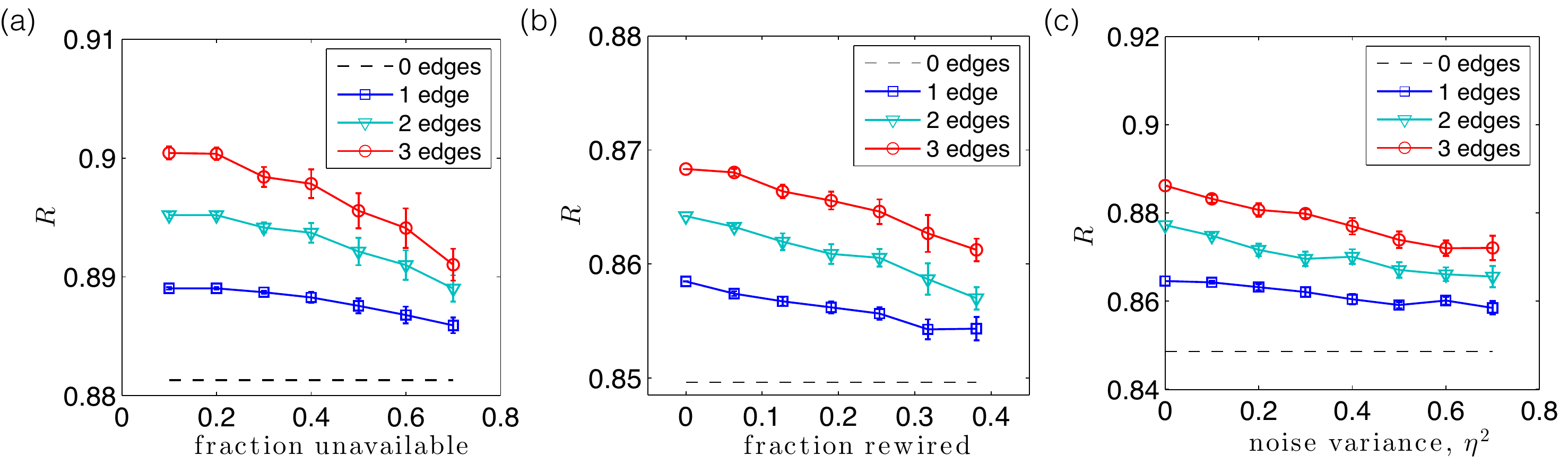}
\caption{ 
Performance of Algorithm~\ref{alg:B} for non-ideal scenarios of synchrony optimization. 
(a)~Dependence of $R$ for a constrained optimization problem in which the edges adjacent to some fraction of the nodes are unavailable and cannot be modified.
(b)~Dependence of $R$ when there is misinformation about the network due to a fraction of the edges being rewired.
(c)~Dependence of $R$ when the natural frequencies have been subjected to Gaussian noise with variance $\eta^2$.
In all panels, curves and error bars indicate the mean and standard error across 10 simulations. 
}
\label{fig:extended}
\end{figure*}

\subsection{Numerical Experiment: Optimization in Non-Ideal Scenarios}\label{sec:unreachable}
Before concluding, we present an extended investigation in which we study the performance of Algorithm~\ref{alg:B} in the following non-ideal situations: 
\begin{itemize}
\item[(a)] when a fraction of the nodes are unavailable in that their edges cannot be perturbed.
\item[(b)] when there is misinformation about the edges that are present in the network;
\item[(c)] when there is misinformation about the natural frequencies.
\end{itemize}
We present results for these respective experiments in Figs.~\ref{fig:extended}(a), (b) and (c). Unless otherwise specified, the natural frequencies are drawn from a normal distribution with unit variance, $K=0.02$, and the network contains $N=50$ nodes and is constructed by the configuration model \cite{bekessy1972asymptotic} with node degrees generated according to a power-law distribution with $\gamma=2.5$, $d_{min}=10$.

In the first study, we investigated a constrained optimization problem in which new edges can only be added to a subset of the nodes---that is, a fraction of the nodes are unavailable for modification. In particular, we select uniformly at random a set of nodes and remove all edges adjacency to them from the set of potential edges $\mathcal{PE}$. We then modify the optimization problem in Sec.~\ref{sec:Maximize} and algorithms of Sec.~\ref{sec:Algorithms} based on this reduced set of potential edges.
In Fig.~\ref{fig:extended}(a), we plot the dependence of $R$ given by Eq.~\eqref{eq:rObjective} after adding edges according to Algorithm~\ref{alg:B} as a function of the fraction of nodes that are unavailable for modification. Note that phase synchronization can be effectively optimized even when a significant fraction of nodes are unavailable to receive new edges.

In the second study, we investigated the effect of misinformation about the network on the performance of synchrony  optimization. That is, rather than implementing Algorithm~\ref{alg:B} using the true network, we used a misinformed network in which a fraction of the edges have been rewired so that there is some discrepancy between the actual network Laplacian $L$ and the one used in the gradient descent algorithm. To construct a misinformed network, we implemented an edge rewiring process in which we iteratively removed an edge and created a new edge uniformly at random from the potential edges. In Fig.~\ref{fig:extended}(b), we plot the dependence of $R$ given by Eq.~\eqref{eq:rObjective} after adding edges according to Algorithm~\ref{alg:B} as a function of the fraction of edges that are rewired. Note that because matrix spectra are relatively robust to perturbations when the eigenvalues are well-spaced \cite{davis1970rotation}, we observe that phase synchrony can still be significantly enhanced even with considerable misinformation about the network structure. 

Finally, in the third study we investigated the effect on algorithm performance when there is misinformation about the natural frequencies. That is, rather than implement Algorithm~\ref{alg:B} using the true natural frequencies, we added to the frequencies $\{\omega_n\}$ Gaussian noise with variance $\eta^2$. In Fig.~\ref{fig:extended}(c), we plot the effect on $R$ for edge additions via Algorithm~\ref{alg:B} as a function of $\eta^2$. Note that the algorithm performs well provided that $\eta^2$ is smaller than the variance of the natural frequencies, which are normally distributed with unit variance, $\sigma_\omega^2=1$.

\section{Discussion}\label{sec:Discussion}

Complex systems exhibiting synchronization are widespread, and for many systems---ranging from the biological rhythms \cite{winfree1967biological} that govern activity in our brains, hearts and other vital organs, to macroscopic systems such  {as power grids}---it is essential that a precise amount of synchronization be present in order to retain proper functionality. For example, a lack of synchronization is a well-known to drive black-outs in power grids \cite{motter2013spontaneous,nishikawa2015comparative,skardal2015control,simonsen2008transient,barahona2002synchronization}, and many neurological tremors are linked to excessive synchronization between neurons \cite{schnitzler2005normal,wilson2015clustered}. 

Here, we explored how to tune and control phase synchronization for network-coupled dynamical systems using network modifications such as adding and/or removing edges. Our analysis is based on recent research \cite{skardal2014optimal} in which we developed a synchrony alignment function (SAF) to measure the interplay between oscillators' heterogeneous natural frequencies and the structural heterogeneity of the network. The SAF is an objective measure for the ability for synchronization to occur for a system with heterogeneous dynamics (e.g., nonidentical natural frequencies $\{\omega_n\}$). Its optimization offers a mathematical framework to design synchrony-optimized systems. Importantly, this approach take into account the actual heterogeneity of the node's dynamics and is complementary to previous research that either lacks or neglects this type of heterogeneity \cite{jalili2008enhancing,jalili2013enhancing,nishikawa2010network}. 

In this research\footnote{ {Note that we have made available several Matlab scripts and a demo to accompany this research at \url{https://github.com/taylordr/SAF_optimization}.}}, we provided the SAF with a more rigorous footing and conducted a spectral perturbation analysis. We derived a first-order expansion that allowed us to approximate how the SAF is effected by network modifications, and this approach is much more computationally efficient than directly recomputing the SAF for the modified system. Specifically, when only a few edges are modified then the approximation is $\mathcal{O}(N^2)$ versus $\mathcal{O}(N^3)$ for recomputing the SAF, where $N$ is the number of oscillators.
By focusing on the addition and removal of edges, we obtained a ranking for the edges (and potential edges) that orders them according to their importance to the SAF and therefore, phase synchronization.  {Importantly, these rankings take into account both the network structure and the heterogeneous oscillator dynamics.} Relying on these rankings, we developed gradient-descent algorithms to efficiently minimize the SAF, which simultaneously maximizes a linear order parameter $R$ that approximates the Kuramoto order parameter $r$. These results complement previous work \cite{skardal2014optimal} where we designed synchrony optimized networks using accept/reject (i.e., Monte Carlo) algorithms. Importantly, here we study a different optimization problem: maximizing phase synchronization using a specified number of edge additions and removals. We showed with numerical experiments (see Fig.~\ref{fig:compare}) that these algorithms significantly outperform other baseline strategies, such as random rewiring or tuning the algebraic connectivity $\lambda_2$, which are \emph{naive} in that they neglect the heterogeneity of oscillator dynamics (i.e., the natural frequencies $\{\omega_n\}$).

The  {theory that we developed here allows us to decide, quantitatively, the extent to which a particular set of connections promote or inhibit phase synchronization and}
can be used to control, engineer and optimize the synchronization properties of complex systems. Our work also provides a mathematical framework with which further optimization techniques can be developed and applied to oscillator networks. It would be interesting to combine the synchrony alignment framework with more advanced optimization techniques  {such as simulated annealing \cite{kirkpatrick1983optimization} and convex optimization \cite{boyd2004convex}. In particular, (by design) gradient-descent algorithms find local optima, not global optima. As previously explored for the optimization of identical oscillators \cite{jalili2013enhancing}, this shortcoming can likely be overcome using, for example, simulated annealing.
It would also be interesting to explore the utility of the SAF for optimizing other aspects of synchronization such as the critical coupling strength at which the phase-locked state appears/disappears, which relates to the quantity $\max_{(i,j)\in\mathcal{E}} |\theta^*_i-\theta^*_j|$ \cite{dorfler2012synchronization}. Synchrony optimization via the SAF minimizes the variance of steady-state phases, and we are currently exploring its utility for tuning the maximum difference.
Finally, it is worth pointing out the rich set of open problems that remain to be tackled, including the dependence of the SAF on various network properties such as the scaling with $N$ and mean degree, degree correlations, clustering, community structure, and so on.}

\appendix

\section{Proof to Proposition \ref{lemma1}}\label{appendixA}

\begin{proof}
We begin with the upper bound.
We will first obtain a relation between ${\|\bm{\theta}-\psi \bm 1\|_2^2}$ and ${\|\bm{\theta}-\overline{\theta} \bm 1\|_2^2}$. We find
\begin{align}
{\|\bm{\theta}-\psi \bm 1\|_2^2} &= || \bm{\theta}-\overline{\theta}\bm 1 + (\overline{\theta}-\psi) \bm 1 ||_2^2 \nonumber\\
&= \langle \bm{\theta}-\overline{\theta}\bm 1 + (\overline{\theta}-\psi) \bm 1,\bm{\theta}-\overline{\theta}\bm 1 + (\overline{\theta}-\psi) \bm 1  \rangle \nonumber\\
&=  {{\|\bm{\theta}-\overline{\theta}\bm{1}\|_2^2}}  + 2\langle \bm{\theta}-\overline{\theta}\bm 1,(\overline{\theta}-\psi) \bm 1 \rangle+ \| (\overline{\theta}-\psi) \bm 1\|_2^2  \nonumber\\
&= N\sigma_\theta^2  +  N |\overline{\theta}-\psi|^2    .\label{eq:bbound0}
\end{align}
Here, the last line uses that the second term vanishes since $\langle \bm \theta - \overline{\theta}\bm 1,\bm1 \rangle=0$. It follows that 
\begin{align}
{\|\bm{\theta}-\psi \bm 1\|_2^2} & \ge {\|\bm{\theta}-\overline{\theta} \bm 1\|_2^2} = N\sigma_\theta^2 . \label{eq:bbound1}
\end{align}
Next, we note that the Kuramoto order parameter is equivalent to the system of equations,
\begin{align}
0&=N^{-1}\sum_{n=1}^N \sin(\theta_n-\psi) \nonumber\\
r& = N^{-1}\sum_{n=1}^N \cos(\theta_n-\psi). \label{eq:OrderParameter2}
\end{align}
We Taylor expand the cosine functions in Eq.~\eqref{eq:OrderParameter2}  {around $0$}, isolate the first two terms, and use Eq.~\eqref{eq:bbound0} to obtain 
\begin{align}
r&= 1- \frac{||\bm{\theta}-\psi \bm{1} ||_2^2}{2N}  +\sum_{k=2}^\infty \frac{(-1)^k ||\bm \theta - \psi \bm 1||^{2k}_{2k}}{(2k)!N}\nonumber\\
&=1-\frac{\|\bm{\theta}-\overline{\theta} \bm 1 \|_2^2 + N|\overline{\theta}-\psi|^2 }{2N}+\sum_{k=2}^\infty \frac{(-1)^k ||\bm \theta - \psi \bm 1||^{2k}_{2k}}{(2k)!N}\nonumber\\
&=R - \frac{ |\overline{\theta}-\psi|^2}{2}+\sum_{k=2}^\infty \frac{(-1)^k ||\bm \theta - \psi \bm 1||^{2k}_{2k}}{(2k)!N}.
\end{align}
Given that the terms in the summation oscillate in sign, our assumption of monotone convergence implies that the summation is  {upper} bounded by the first term, $||\bm \theta-\psi \bm 1||_4^4/(4!N)$. Combining this bound with Eq.~\eqref{eq:bbound1} recovers the upper bound in Eq.~\eqref{eq:equivalence}.  
We next prove the lower bound. Monotone convergence also implies that the summation is positive, which gives the lower bound
\begin{align}
r&\ge R - \frac{ |\overline{\theta}-\psi|^2}{2}. \label{eq:bound3}
\end{align}
To bound the difference between the mean fields, $\psi$ and $\overline{\theta}$, we Taylor expand the sine functions in Eq.~\eqref{eq:OrderParameter2}, isolate the first term, and rearrange to obtain
\begin{align}
\overline{\theta}-\psi  &=   \sum_{k=1}^\infty   \frac{(-1)^{k+1}}{(2k+1)!N} \sum_{n=1}^N(\theta_n-\psi)^{2k+1} .
\end{align}
Note that terms in the summation oscillate in sign so that terms $k=1,3,\dots$ have the same sign as $ \overline{\theta}-\psi $. Under our assumption of monotone convergence, the magnitude of the summation is bounded by the magnitude of the first term. We neglect the remaining terms and take the absolute value of both sides to obtain Eq.~\eqref{eq:meanfields}.
\end{proof}

\section{Proof to Theorem \ref{eq:stationarySolution}}\label{appendixB}
\begin{proof}
In the state of phase-locked synchronization, $d \theta_n /dt = \Omega$ for every oscillator so that Eq.~\eqref{eq:HLD} becomes 
\begin{align}\label{eq:Omega1}
\Omega \bm{1}=\bm{\omega}- K{L}\bm{\theta}^*.
\end{align} 
The Moore-Penrose inverse  {$L^{\dagger}=\sum_{n=2}^{N}\lambda_n^{-1}\bm{v}^{(n)}\bm{v}^{(n)\top}$} is defined so that $L^\dagger L^\dagger L = L $ and $L^\dagger LL^\dagger = L^\dagger $.  {Recall that the eigenvectors $\{\bm{v}^{(n)}\}$ of $L$ define an orthonormal basis, and our assumption of a connected undirected network implies $0=\lambda_1<\lambda_2\dots\leq\lambda_N$.}
We multiply both sides by $K^{-1}L^\dagger$ to obtain a general solution of the form
\begin{align}
\bm{\theta}^* =  K^{-1}{L}^\dagger \bm{ {\omega}}  -  K^{-1} {L}^\dagger (\Omega \bm1) + c\bm v^{(1)}, \label{eq:qwe1}
\end{align}
where $\bm v^{(1)}=N^{-1/2}\bm 1$ is the eigenvector corresponding to the trivial eigenvalue $\lambda_1=0$ and $c\in\mathbb{R}$ is a constant that accounts for the projection of $\bm{\theta}^*$ onto the nullspace,
$\text{null}(L^\dagger)=\text{null}(L)=\text{span}(\bm v^{(1)}).$
Because $\bm1\in\text{null}(L^\dagger)$, $L^\dagger (\Omega\bm 1)=0$ and the second term vanishes.
To solve for $c$, we multiply both sides of Eq.~\eqref{eq:qwe1} by $N^{-1}\bm 1^T$ to obtain $c=N^{1/2}\overline{\theta}$  {(i.e., $c\bm v^{(1)} = \overline{\theta}\bm 1$).}
To complete the proof, we use Eq.~\eqref{eq:stationarySolution} to obtain
\begin{align}
R &= 1-\sigma_\theta^2/2 \nonumber\\
&= 1- \frac{1}{2N}|| \bm{\theta}^*-\overline{\theta} \bm{1}||^2 \nonumber\\
&=1-\frac{1}{2N}|| K^{-1}{L}^\dagger \bm{ {\omega}}   ||^2 \nonumber\\
&=1-J(\bm{ {\omega}},  L)/2K^2 .
\end{align}
\end{proof}

\section{Proof to Theorem \ref{sec:General}}\label{appendixC}
\begin{proof}
We define
\begin{align}
F(\epsilon) = J(\bm{ {\omega}},  L +\epsilon \Delta L) = \frac{1}{N}\sum_{n=2}^N \frac{f_n(\epsilon)} {g_n(\epsilon) },
\end{align}
where
$f_n(\epsilon) = [ \bm{\omega}^T\bm{v}^{(n)} (\epsilon)]^2$ and $g_n(\epsilon) = \lambda_n^{2}(\epsilon),$ 
and we seek a solution of the form 
\begin{align}
F(\epsilon) = F(0) + \epsilon F'(0)+ \mathcal{O}(\epsilon^2).
\end{align} 
Here, we use $F'(\epsilon)$ to denote the derivative with respect to $\epsilon$, $F'(\epsilon)=\frac{dF}{d\epsilon}$. Using the quotient rule, we find  
\begin{align}
F'(\epsilon)
 &=  \frac{1}{N} \sum_{n=2}^N \frac{f'_n(\epsilon)g_n(\epsilon) - f_n(\epsilon)g'_n(\epsilon)}{g_n^2(\epsilon)},
\end{align}
where $f_n'(\epsilon) = 2[ \bm{\omega}^T\bm{v}^{(n)} (\epsilon)][ \bm{\omega}^T{\bm{v}^{(n)}}' (\epsilon)]$ and $g_n'(\epsilon) = 2\lambda_n(\epsilon)\lambda_n'(\epsilon)$. 
Evaluation of this expression at $\epsilon=0$ yields
\begin{align}
F'(0)
 &=  \frac{1}{N} \sum_{n=2}^N \frac{ 2[ \bm{\omega}^T\bm{v}^{(n)}][ \bm{\omega}^T{\bm{v}^{(n)}}' ] \lambda_n^{2}}{\lambda_n^4}
  - \frac{[ \bm{\omega}^T\bm{v}^{(n)} ]^2\left[2\lambda_n \lambda_n' \right]}{\lambda_n^4},
  \label{eq:eval_at_zero}
\end{align}
where we have dropped the argument $\epsilon$ when $\epsilon=0$ to simplify our presentation. Recall that $\lambda'_n=(\bm v^{(n)})^T \Delta L \bm v^{(n)}$ and ${\bm{v}^{(n)}}' =\sum_{{m\not=n}} \frac{ (\bm v^{(m)})^T \Delta L \bm v^{(n)}}{\lambda_n-\lambda_m } \bm v^{(m)}$ are well-known perturbation results given in Eq.~\eqref{eq:First2}. We substitute these results into Eq.~\eqref{eq:eval_at_zero} and combine terms to obtain
\begin{align}
F'(0)
 &=  \frac{2}{N}\sum_{n=2}^N   \left(\frac{ \bm{\omega}^T  \bm v^{(n)} }{   \lambda_n ^3}  \right)
\left(\sum_{{m=1}}^N \frac{ [ \bm{\omega}^T\bm v^{(m)}][(\bm v^{(m)})^T \Delta L \bm v^{(n)}]}{(1-\frac{\lambda_m}{\lambda_n}) -\delta_{nm}}   \right) .
\end{align}

\end{proof}

\section{Proof to Corollary \ref{Corr_11}}\label{appendixD}
 \begin{proof}
We first note that $\epsilon=1$ for the modification of an unweighted edge. 
Given a Laplacian matrix $L$, the new Laplacian matrix after adding or removing an undirected edge $(p,q)$ has the form $L'=L+ \Delta L^{(pq)} $ or $L'=L- \Delta L^{(pq)} $, respectively, where $\Delta L_{ij}^{(pq)} $ is given by Eq.~\eqref{eq:DL} Using Eq.~\eqref{eq:DL}, it is straightforward to show
\begin{align}
(\bm v^{(m)})^T \Delta L^{(pq)} \bm v^{(n)} &= (\bm v^{(m)}_p-\bm v^{(m)}_q)(\bm v^{(n)}_p-\bm v^{(n)}_q).
\end{align}
We substitute this result into Eq.~\eqref{eq:First_J} to complete the proof.
\end{proof}

\section{Proof to Corollary \ref{Corr_12}}\label{appendixE}
\begin{proof}
Due to linearity, it follows that 
\begin{align}
\Delta L=  \sum_{(p,q)\in\mathcal{E}^{(+)}}  \Delta L^{(pq)} -\sum_{(p,q)\in\mathcal{E}^{(-)}}  \Delta L^{(pq)} . 
\end{align}
Thus
\begin{align}
(\bm v^{(m)})^T \Delta L \bm v^{(n)} &= \sum_{(p,q)\in\mathcal{E}^{(+)}}  ( \bm v^{(m)}_p-\bm v^{(m)}_q)(\bm v^{(n)}_p-\bm v^{(n)}_q)  \nonumber\\ &-
\sum_{(p,q)\in\mathcal{E}^{(-)}}  ( \bm v^{(m)}_p-\bm v^{(m)}_q)(\bm v^{(n)}_p-\bm v^{(n)}_q).
\end{align}
We substitute this result into Eq.~\eqref{eq:First_J} and simplify to recover Eq.~\eqref{eq:dJ_2edges}.
\end{proof}

\bibliographystyle{siam}
\bibliography{SAF_perturb}

%
\end{document}